\documentclass{article}

\usepackage{arxiv}

\usepackage[utf8]{inputenc} % allow utf-8 input
\usepackage[T1]{fontenc}    % use 8-bit T1 fonts
\usepackage[hidelinks]{hyperref}       % hyperlinks
\usepackage{url}            % simple URL typesetting
\usepackage{booktabs}       % professional-quality tables
\usepackage{amsfonts}       % blackboard math symbols
\usepackage{nicefrac}       % compact symbols for 1/2, etc.
\usepackage{microtype}      % microtypography
\usepackage{graphicx}
\usepackage{natbib}
\usepackage{doi}
\usepackage{balance} % for balancing columns on the final page
\usepackage{multirow}
\usepackage{amsthm}
\usepackage{amsmath}
\usepackage{enumitem}
\usepackage{xcolor}
\usepackage[english]{babel}

\usepackage{times}
\usepackage{soul}
\usepackage{caption}
\usepackage{amsthm}
\usepackage{array}
\usepackage[switch]{lineno}

\usepackage{acronym}
\usepackage{bbm}
\usepackage{multicol}
\usepackage{subcaption}

\usepackage{amsmath}       % for math environments and symbols
\usepackage{amssymb}       % for more math symbols
\usepackage{algorithm}
\usepackage{algpseudocode} 
\usepackage{bm}            % for bold math symbols
\usepackage{graphicx}      % for figures, if needed
\usepackage{hyperref}      % for clickable refs (optional)
\usepackage{caption}       % for caption 
\usepackage[normalem]{ulem}
\usepackage[most]{tcolorbox}
\algrenewcommand\algorithmiccomment[1]{\hfill \textcolor{gray}{// #1}}

\renewcommand{\paragraph}[1]{\noindent\textbf{#1}}

\newcommand{\ourrealgo}{\textsc{ROSA}}

\newcommand{\ourmabalgo}{\textsc{TRCM-UCB$_\texttt{OPT}$}}
\newcommand{\ourlinucb}{\textsc{REV-BaseLinUCB-S}}
\newcommand{\oursupucb}{\textsc{REV-SupLinUCB-S-OPT}}

\theoremstyle{mythmstyle}
\newtheorem{theorem}{Theorem}
\newtheorem{proposition}{Proposition}
\newtheorem{claim}{Claim}

\newtheorem{definition}{Definition}
\newtheorem{lemma}{Lemma}
\newtheorem{corollary}{Corollary}

\algnewcommand{\Input}{\textbf{Input: }}
\algnewcommand{\Output}{\textbf{Output: }}
\algnewcommand{\Parameter}{\textbf{Parameter: }}
\algnewcommand{\Wp}{\textbf{with probability }}
\algnewcommand{\Elsecon}{\textbf{else }}

\tcbset{
  mechbox/.style={
    colback=gray!5,      % very light grey background
    colframe=gray!40,    % light grey border
    fonttitle=\bfseries,
    coltitle=black,
    boxrule=0.4pt,
    arc=2pt,
    outer arc=2pt,
    left=6pt,
    right=6pt,
    top=4pt,
    bottom=4pt,
    boxsep=2pt,
  }
}

\title{Truthful Reverse Auctions for Adaptive Selection via Contextual Multi-Armed Bandits%
\thanks{Appears as a Full Paper in the Proceedings of the 25th International Conference on Autonomous Agents and Multiagent Systems (AAMAS 2026).}}

\author{ {\hspace{1mm}Pronoy Patra} \\
	International Institute of Information \\
        Technology (IIIT), Hyderabad, India \\
	\texttt{pronoy.patra@research.iiit.ac.in} \\
	\And
	{\hspace{1mm}Sankarshan Damle} \\
	Microsoft Research India \\
	\texttt{t-sandamle@microsoft.com} \\
	\And
        {\hspace{1mm}Manisha Padala} \\
	Indian Institute of Technology,\\ 
        Gandhinagar, India \\
	\texttt{manisha.padala@iitgn.ac.in} \\
        \And
        {\hspace{1mm}Sujit Gujar} \\
	International Institute of Information \\
        Technology (IIIT), Hyderabad, India \\
	\texttt{sujit.gujar@iiit.ac.in} \\
}

\date{}

\renewcommand{\headeright}{}
\renewcommand{\undertitle}{}
\renewcommand{\shorttitle}{Truthful Reverse Auctions for Adaptive Selection via Contextual Multi-Armed Bandits}

\hypersetup{
pdftitle={Truthful Reverse Auctions for Adaptive Selection via Contextual Multi-Armed Bandits},
pdfsubject={cs.GT, econ.TH},
pdfauthor={Pronoy ~Patra, Sankarshan ~Damle, Manisha ~Padala, Sujit ~Gujar},
pdfkeywords={Contextual MAB, Auction Design, Optimal Auction, Learning},
}

\begin{document}
\maketitle
\begin{abstract}
We study the problem of selecting large language models (LLMs) for user queries in settings where multiple LLM providers submit the cost of solving a query. From the user’s perspective, choosing an optimal model is a sequential, query-dependent decision problem: high-capacity models offer more reliable outputs but are costlier, while lightweight models are faster and cheaper. We formalize this interaction as a reverse auction  design problem with contextual online learning, where the user adaptively discovers which model performs best while eliciting costs from competing LLM providers. Existing multi-armed bandit (MAB) mechanisms focus on forward auctions and social welfare, leaving open the challenges of reverse auctions, provider-optimal outcomes, and contextual adaptation. We address these gaps by designing a resampling-based procedure that generalizes truthful forward MAB mechanisms to reverse auctions and prove that any monotone allocation rule with this procedure is truthful. Using this, we propose a contextual MAB algorithm that learns query-dependent model quality with sublinear regret. Our framework unifies mechanism design and adaptive learning, enabling efficient, truthful, and query-aware LLM selection.
\end{abstract}

% keywords can be removed

%%%%%%%%%%%%%%%%%%%%%%%%%%%%%%%%%%%%%%%%%%%%%%%%%%%%%%%%%%%%%%%%%%%%%%%%
\section{Introduction}
\label{sec:intro}
% motivating example
\emph{Large Language Models} (LLMs)~\cite{achiam2023gpt} increasingly shape the way individuals and organizations make decisions, collaborate, and innovate, reflecting their growing role as general-purpose AI systems~\citep{li2023camel,microsoft2023copilot,google2023bardupdate,su-etal-2024-language}. Beyond single-turn tasks, recent work deploys LLMs as \emph{agentic} systems that plan, act, and interact with tools or other LLMs to complete complex objectives~\cite{park2023generative,zhao2023depth}. In both single-agent and multi-agent settings, \emph{LLM Providers} (e.g., OpenAI~\cite{openai2025} or Anthropic~\cite{anthropic2025}) accelerate adoption by releasing multiple model variants that explicitly trade off capability, latency, and cost~\cite{anthropic2025claude4,openai2025gpt5}. Concretely, modern LLM deployments have established the problem of selecting an appropriate model for a given query. For example, GPT-5 employs a real-time router that dynamically allocates queries between a fast base model and a deeper reasoning model, leveraging signals such as task complexity, tool requirements, and explicit user intent~\cite{openai2025gpt5}. However, these routing scenarios focus on internal model orchestration within a single provider, where the user may or may not have control over model selection or cost optimization.

% \sg{i guess, we cud be more explicit about that the above work try to optimizae which model to be used..but assumes costs to be known..typically private...also if none of them uses contextual MAB..point out that...if none of them consider users utility , point out that these are shortcomings in existance literature..no need to use word shortcomings..but while reading message should be that...if these things are taken care, overall intro would look good} \sankarshan{it is contextual, they would also trade off utility - cost for assisgnments, but not in the mechanism design sense} \sankarshan{I added a line --}

%% Paper goal
\paragraph{Our Goal.}
From the perspective of a user, this trade-off remains central: the user derives \emph{utility} from the model’s performance while incurring a \emph{cost} for using it. Intuitively, high-capacity models offer more accurate or reliable outputs but are more expensive, whereas lightweight models are faster and cheaper but may underperform on complex tasks. By framing the problem from the user’s perspective, we aim to study how a user trades off the quality of a model’s output against its usage cost. At the same time, we account for the LLM provider’s ability to elicit the cost of solving a query. This formulation provides a foundation for analyzing optimal model selection strategies for the user and the design of mechanisms by the provider.

%% Addresss -- set up the "elicitation" and "learning"
% To formalize this, we focus on two complementary challenges: \textit{elicitation} and \textit{learning}. \sg{If, for a querry, expected performance of available LLMs is known, the user faced problem of elicitation.} Elicitation corresponds to a reverse auction mechanism design problem~\cite{} where multiple LLM providers submit costs for solving a user query, which may differ across models and tasks. \sg{In mechanism design, the goal could be maximize social welfare or the users utility. The later one is commonly known as \emph{optimal auction design}~\cite{myerson1981optimal}. If the performance of the models is known, designing reverse optimal auction is straightforward.} \sg{However, the performance is stochastic, and initially, the user may not be aware and need to learn. Thus,} Learning corresponds to an online, contextual optimization problem from the user’s perspective, where the user adaptively discovers which model is likely to solve a given query most effectively. Together, these perspectives capture the interaction between provider-side cost elicitation and user-side query-dependent model selection, providing a foundation for designing mechanisms that are both truthful and efficient.
To formalize this, we focus on two complementary challenges: \textit{elicitation} and \textit{learning}. When the expected performance of available LLMs for a given query is known, the user faces an elicitation problem. Researchers model this as a \emph{reverse auction} where providers submit costs for processing the query. In mechanism design, the objective may be to maximize social welfare or the buyer’s (user’s) utility, the latter corresponding to \emph{optimal auction design}~\cite{myerson1981optimal}. If model performance were known deterministically, designing a user-optimal reverse auction would be straightforward. However, model performance is stochastic and context-dependent; the user must learn it over time. This introduces a complementary {learning} challenge, formalized as a (contextual) sequential optimization problem, where the user adaptively identifies which provider performs best for each query. Together, these perspectives unify provider-side cost elicitation with user-side adaptive selection, forming the basis for truthful and efficient LLM allocation mechanisms.

\paragraph{Challenges.}
Addressing the combined elicitation and learning problem poses several challenges. Generally, researchers use multi-armed bandit (MAB) methods~\cite{bubeck2012regret} as a standard framework for adaptive or sequential optimization problems. In the context of mechanism design, existing work on MAB mechanisms primarily focuses on traditional forward auctions that measure efficiency in terms of social welfares~\cite{babaioff2009characterizing,babaioff2015truthful}. These approaches leave two important gaps: (i) no prior work studies MAB mechanisms for reverse auctions, where the user acts as a buyer and multiple LLM providers bid to solve a query, and (ii) no prior work studies MAB mechanisms for optimal auctions that aim to maximize provider utility rather than social welfare. Moreover, to our knowledge, except for \cite{abhishek2020designing}, no prior work incorporates contextual information into these mechanisms. Our work addresses these gaps by designing a mechanism that is both truthful and optimal from the LLM provider’s perspective, while enabling users to adaptively learn which model solves a query most effectively via a contextual MAB framework.

%% Challenges
    % \item no MAB mechanism for reverse auctions
    % \item no MAB mechanism for optimal auction\\ existing works for social welfare maximizing MAB auctions
    % \item barring Abhishek's paper, no contextual MAB mechanism 
    % \item ours is first to design a truthful reverse optimal auction with learning via a contextual MAB algorithm 

%% Approach & Contributions
% \subsection*{Our Approach \& Contributions}
% \sg{\cite{babaioff2009characterizing},\cite{devanur2009} showed that any deterministic auciton that deploys MAB algorithm for learning stochatic parameters incurs much higher regret than the one without involving auction. Towards this, \cite{} proposed interesting randomization procedure that ensures truthfulness and retains the regret guartees of the underlying MAB algorithm. \cite{} designed randomization, called \emph{resampling procedure} wherein the bids of the agents are updated randomly while ensuring truthfulness. This procedure works as long as the MAB algorithm is monotone. It should be noted that, the procedure designed is for forward auction and social welfare maximizing auction. In this work, we draw lines parallel to \cite{}, and design resampling procedure that works for reverse auction and for the goal of optimal auction. Once, such procedure is available, the user's problem remain is to design contextual MAB algorithm that ensures monotonicity, i.e., decreasing the reported cost should not decrease the allocation. To this end, we design XXXXX, a contextual MAB learning algorithm for reverse auction and ensures truthfulness and incurs regret of $O()$}.

\paragraph{Our Approach.}
\citet{babaioff2009characterizing} and \citet{devanur2009mmab} showed that any deterministic auction that directly embeds a MAB algorithm to learn stochastic parameters suffers substantially higher regret compared to its non-strategic counterpart. To address this, \citet{babaioff2009characterizing} introduced a randomized \emph{resampling procedure} that preserves truthfulness while retaining the regret guarantees of the underlying MAB algorithm. The key idea is to randomize  bids 
% reminder that 'bids' needs to be somewhere earlier
in a manner consistent with Bayesian incentive compatibility (BIC), provided that the MAB algorithm itself is monotone in allocations. However, this framework is limited to \emph{forward} auctions optimizing social welfare. In contrast, we develop an analogous randomized adjustment procedure tailored to \emph{reverse} auctions, where the objective is user-optimal (cost-minimizing) rather than welfare-maximizing. Once such a truthful reverse procedure is available, the remaining challenge lies in designing a contextual MAB algorithm that preserves monotonicity: ensuring that lowering the reported cost does not reduce allocation probability. To this end, we propose \ourmabalgo: a \underline{t}ruthful, provider-optimal, \underline{r}everse \underline{c}ontextual \underline{M}AB algorithm that achieves truthfulness and an $O(\sqrt{T})$ regret bound in reverse, stochastic settings.

\paragraph{Our Contributions.} \textbf{First,} we formalize the problem of selecting among competing LLM providers as a \emph{optimal reverse contextual MAB mechanism}, where the user solicits costs from providers and adaptively learns which model best solves each query. \textbf{Second,} we adapt the theory of optimal auctions to this reverse setting, showing how provider-side incentives can be aligned with user-side selection in our context. \textbf{Third,} we design a novel resampling procedure (Algorithm~\ref{alg:rosa}), inspired by \citet{babaioff2015truthful} but tailored to reverse, provider-optimal settings rather than forward, social welfare–maximizing ones. \textbf{Fourth,} we prove that any monotone allocation rule, combined with our resampling procedure, yields a truthful reverse MAB auction (Theorem~\ref{thm:truthful_epic_epir}). \textbf{Last,} we instantiate this result with \ourmabalgo\ (Algorithm~\ref{alg:suplinucb}) that leverages contextual information to learn model quality over time, achieving $O(\sqrt{T})$ regret across $T$ queries. Together, these results establish the first truthful, reverse, contextual MAB auction for adaptive LLM model selection.

% %%%%%%%%%%%%%%%%%%%%%%%%%%%%%%%%%%%%%%%%%%%%%%%%%%%%%%%%%%%%%%%%%%%%%%%%

\section{Related Work}
\label{sec:rel_wk}

Here, we review prior work on mechanism design and MABs.%, highlighting gaps relevant to reverse, provider-optimal LLM selection. 
% We first discuss canonical results in MABs and reverse auctions, and then turn to recent research at the intersection of MABs, auctions, and LLMs.

\paragraph{Multi-Armed Bandits (MABs)}
The multi-armed bandit (MAB) problem~\cite{robbins1952some,bubeck2012regret} is a fundamental framework for sequential decision-making under uncertainty. A lot of variants of Multi-Armed bandits have been made recently, viz, combinatorial~\cite{chen2013combinatorial, solanki2024fairness, solanki2023dpfcbandits}, budgeted~\cite{das2022budgeted}, sleeping~\cite{abhishek2021sleeping}, volatile~\cite{kumar2025regret}, etc. Contextual bandits~\cite{li2010contextual,abbasi2011improved} further extend this paradigm by modeling reward distributions as functions of observed contextual information. Applications of MABs span diverse domains such as federated learning, tariff design~\cite{chandlekar2023tariff}, and constrained subset selection~\cite{deva2021subset}. 

A wide range of solution strategies has been developed, most notably Upper Confidence Bound (UCB) methods and Thompson Sampling.

%Recent work studies bandit-based approaches for cost-efficient multi-LLM selection for different kinds of rewards~\cite{dai2024cost}. However, they do not assume that the LLMs are strategic and may misreport costs.

\paragraph{Reverse Auctions}
We model the LLM provider–user interaction as a reverse auction, building on the classic theory of optimal auctions and mechanism design~\cite{myerson1981optimal,maskin1984monopoly}. In a procurement or reverse auction, the buyer (the user, in our setting) seeks to minimize cost while ensuring truthful reporting by the sellers (the LLM providers). Several works extend Myerson’s framework~\cite{myerson1981optimal} to procurement settings~\cite{cary2008auctions} and to budget-feasible mechanisms~\cite{singer2010budget}. However, contrary to our setup, these mechanisms generally assume deterministic valuations.

\paragraph{LLMs and Multi-Armed Bandits} 
Recent advances in LLMs motivate the need for model routing and selection, where a user or system dynamically decides which model to invoke for a given query. Early work focused on static model choice based on latency and cost trade-offs~\cite{hu2024routerbench,yuan2025who}, while recent methods use dynamic routers that predict query complexity and select among heterogeneous models~\cite{openai2025gpt5,anthropic2025claude4,park2023generative,feng2024graphrouter}. A new line of research also studies the intersection of LLMs with MABs~\cite{dai2024cost,sun2025large,poon2025online}.
\citet{dai2024cost} formulate the multi-LLM selection problem as a cost-aware MAB, optimizing performance–cost trade-offs across heterogeneous LLMs under diverse reward models. \citet{poon2025online} introduces the first contextual bandit framework for adaptive LLM selection under unstructured, dynamically evolving prompts, achieving sublinear regret while balancing accuracy and cost.
Our formulation complements this line of work by introducing an incentive-aware framework: rather than training a router to minimize latency or maximize accuracy, we design a mechanism that elicits truthful costs from LLM providers.
% Other approaches train controller models to automatically select the most cost-efficient LLM~\cite{li2024adaptive,ye2024mixture}.

\paragraph{MABs and Mechanism Design} In many real-world decision-making scenarios, mechanism design and online learning intersect when agents’ private information influences the outcome of a learning process. Specifically, in multi-armed bandit (MAB) settings with strategic agents the learner must balance exploration and exploitation while eliciting truthful information about privately known parameters that impact reward estimation or allocation. However, since these parameters are self-reported, agents may strategically misrepresent them to manipulate the learning process or improve their expected payoffs. To ensure truthful participation while maintaining efficient learning, MAB-based mechanism design integrates incentive-compatible auction mechanisms into the learning framework~\cite{babaioff2009characterizing,babaioff2015truthful,devanur2009mmab}. 
Classical formulations consider \emph{forward auctions}, where an auctioneer repeatedly allocates resources to agents with private values, aiming to maximize social welfare. \citet{babaioff2009characterizing} introduced a resampling-based procedure to achieve Bayesian incentive compatibility (BIC) in stochastic MAB environments, a result later extended and refined in~\cite{babaioff2015truthful,devanur2009mmab,JAIN201844}. Subsequent work has explored combining bandit learning with auction design across domains such as crowdsourcing~\cite{jain2016deterministic},  smart grids~\cite{shweta2020aaai}, sponsored search auctions~\cite{gao2021auction}, and bidding strategy optimization for displaying ads~\cite{guo2024bayesian}, multi-unit procurement~\cite{bhat2019bidimensional}. However, these mechanisms primarily optimize {social welfare} and operate in forward auction settings. To our knowledge, only \citet{abhishek2020designing} incorporate contextual information into MAB mechanisms, still within a welfare-maximizing framework. In contrast, our work introduces the first reverse optimal MAB mechanism, aligning provider-side incentives with user-side adaptive learning while ensuring truthfulness and individual rationality.
In summary, no prior work studies reverse, contextual MAB mechanisms, where the user acts as a buyer eliciting private costs from multiple LLM providers acting as sellers.

% Our work extends this literature by (i) introducing the first **reverse, provider-optimal** MAB mechanism, and (ii) proposing a **resampling procedure** tailored to this setting, ensuring Bayesian incentive compatibility and individual rationality while allowing user-side contextual learning.

%%%%%%%%%%%%%%%%%%%%%%%%%%%%%%%%%%%%%%%%%%%%%%%%%%%%%%%%%%%%%%%%%%%%%%%%
\section{Preliminaries}
\label{sec::prelims}
We now present the necessary background for our model, and cover the key components of reverse auctions and multi-armed bandit (MAB) mechanisms that underpin our approach. 
Table~\ref{tab::notations} tabulates the notations used in this paper.

\begin{table}[t]
\centering
\caption{Notations~\label{tab::notations}}
\begin{tabular}{ll}
\toprule
\textbf{Symbol} & \textbf{Description} \\
\midrule
$T$ & Total number of rounds \\
$t$ & Round index, $t \in \{1, \dots, T\}$ \\
$M$ & Number of LLM providers \\
$i$ & Provider index, $i \in [M]$ \\
$q_t$ & Query from the user at round $t$ \\
$x_t \in \mathbb{R}^d$ & Context vector associated with query $q_t$\\
$c_i$ & True private cost of provider $i$ \\
$c$ & Cost vector $(c_1, \dots, c_M)$ \\
$b_i$ & Bid submitted by provider $i$ \\
$F_i(\cdot), f_i(\cdot)$ & CDF and PDF of $i$'s cost distribution \\
$\theta_i \in \mathbb{R}^d$ & Unknown parameter vector for provider $i$ \\
$r_{i,t}$ & Realized stochastic reward from provider $i$ at round $t$ \\
$v_{i,t}$ & True expected reward: $\theta_i^\top x_t$ \\
$\hat{v}_{i,t}$ & Estimated valuation for provider $i$ at round $t$ \\
$\mathcal{A}_{i,t}(c)$ & Allocation; 1 if $i$ is chosen at $t$, else 0 \\
$p_{i,t}(c)$ & Payment to provider $i$ at round $t$ \\
$u_{i,t}(c)$ & Utility: $p_{i,t}(c) - c_i \cdot \mathcal{A}_{i,t}(c)$ \\
$\Psi_i(c_i)$ & Virtual cost: $c_i + \frac{F_i(c_i)}{f_i(c_i)}$ \\
$z_i(c_{-i})$ & Critical cost threshold for $i$ given $c_{-i}$ \\
$h_{t-1}$ & History up to the start of round $t$ \\
$I_t$ & Provider selected at round $t$ by the mechanism \\
$i^*_t$ & Oracle-optimal provider at round $t$ \\
$\mathbb{R}_T(\mathcal{ALG})$ & Cumulative regret up to horizon $T$ \\
\bottomrule
\end{tabular}
\end{table}

%%%%%%%%%%%%%%%%%%%%%%%%%%%%%%%%%%%%
\subsection{Model}
%%%%%%%%%%%%%%%%%%%%%%%%%%%%%%%%%%%%
% \sg{E.g., you may want to start like there is a query $x^t$ at round $t$ and there are $[M] = \{1,2,\ldots,M\}$ LLM provides available to get the answer from. The user obtains a valuation of $v_i^t$ if the query is answered from LLM $i$. The LLM needs to be compensated for its cost $c_i^t$ for answering the query $x^t$. The $c_i^t$s are private to the respective LLMs. Let $p_i^t$ be the payment of $i$ at round $t$, if its service is requested. Thus, the utility to the user is $u^t = \sum_{i\in [M]}(v_i^t - p_i^t)\cdot {\Chi_i}^t$ where ${\Chi_i}^t$ is indicator variable denoting if $i$'s service was obtained or no. }\sg{if you have used any other indicator variable, I am fine..no need to stick to ${\Chi_i}^t$.}
% \sg{The user's goal is to maximize $u^t$.  LLM $i$'s utility $u_i^t = \Chi_i^t\cdot(p_i^t - c_i^t)$. We need to ensure that the LLMs provide $c_i^t$ truthfully, which is captured through \emph{incentive compatibility}. The overall problem in mechanism design theory is referred to as designing an \emph{optimal reverse auction}. First, we explain how to design an optimal truthful mechanism with the required definitions from game theory.}
% \sankarshan{@Pronoy: you should check if the above comments are taken care of.}

\paragraph{Setup}
We consider a user $\mathfrak{U}$ who submits a sequence of queries over rounds $t = 1, 2, \ldots, T$. 
At each round $t$, the user issues a query $q_t$, and a set of strategic LLM providers, denoted by $[M] = \{1, 2, \ldots, M\}$, compete to process it. 
Each provider may differ in capability, latency, and cost, reflecting heterogeneous model architectures or inference configurations~\cite{anthropic2025claude4,openai2025gpt5}.

\paragraph{Provider Costs and Bids}
Each provider $i \in [M]$ has a true private cost $\zeta_i$ per token, representing its internal computation cost. For a query $q_t$, provider $i$ would required $n_{i,q_t} \in \mathbb{N}$ tokens, the actual computation or latency cost incurred by provider $i$ is $c_{i,t} \;=\; \zeta_i \cdot n_{i,q_t}$.

The strategic provider submits a bid $b_{i,t}$ indicating the price to process the query (the provider submits bids per token, it can be easily converted to bid for the query). The goal for the provider would be to earn as much as possible.
In general, the desirable goal in mechanism (or auction) design is to propose (i) an allocation and (ii) a payment rule that incentivizes providers to bid $b_{i,t} = c_{i,t}$.

\paragraph{Allocation ($\mathcal{A}$) and Payments ($\mathcal{P}$)}
At any round $t$ $\in [T]$, the mechanism determines an allocation $\mathcal{A}_t(b_t)$
\[
   \mathcal{A}_t(b_t) \;=\; \{ \mathcal{A}_{1,t}(b_t), \ldots, \mathcal{A}_{M,t}(b_t) \}, 
   \qquad \sum_{i\in[M]} \mathcal{A}_{i,t}(b_t) \leq 1,
\]
where $\mathcal{A}_{i,t}(b_t) \in \{0,1\}$ where $1$ denotes that bidder/LLM $i$ is selected given reported bids $b_t = (b_{1,t},\ldots,b_{M,t})$. If 
provider $i$ is selected, the user makes a payment $p_{i,t}(b_t)$ to $i$ and others do not receive any payments. %\sankarshan{the others pay zero?}. 
Let $\mathcal{P}_t: p_t(b_t) =(p_{1,t}(b_t),\ldots,p_{M,t}(b_t))$ capture the payment rule. 

\paragraph{Query Context and User Valuation}
At round $t$, query $q_t$ is associated with observable features $x_t \in \mathbb{R}^d$, 
which may include query embeddings or user-specific attributes (e.g., OpenAI's saved memory feature~\cite{openai2024memory}). {Henceforth, we identify query $q_t$ with $x_t$. }For 
provider $i$, the expected value to the user of receiving a satisfactory 
response is modeled as
\[
   v_{i,t} 
   \;=\; \mathbb{E}[r_{i,t} \mid x_t] 
   \;=\; \theta_i^\top x_t,
\]
where $\theta_i \in \mathbb{R}^d$ denotes a parameter vector specific to 
provider $i$, and $r_{i,t}$ denotes the realized response quality for query $x_t$. %\sankarshan{can $x$ not directly depend on $r$, $x$?}

\paragraph{Utilities}
Let the user's total utility at round $t$ be, 
\[
   u_{0,t} 
   \;:=\;  \sum_{i \in [M]} 
    \mathcal{A}_{i,t}(b_t) \cdot \big( v_i - p_{i,t}(b_t) \big).
\]
%\sankarshan{the "0" is confusing. I would prefer $\mathbf{u}_t$ (bold)}\sg{In auctions, referring auctioneer as 0 is common. i am fine with what you proposed too.}
The utility of provider $i$ is given by
\begin{equation}   
   u_{i,t} 
   \;:=\; 
   \big( p_{i,t}(b_t) - c_{i,t} \big) \cdot \mathcal{A}_{i,t}(b_t)
\end{equation}

\paragraph{Note}
If $\theta_i$s are known, $\mathfrak{U}$'s objective at round $t$ is designing an optimal reverse auction. For completeness, we show how it can be done in the next subsection.  For ease of exposition: (i) we assume $n_{i,q_t}$ to be same for all queries. That is, $n_i = n_{i,q_t}\forall q_t$. Note that, as explained later, the results can be trivially extended for the case when $n_{i,q_t}$s are different. With this assumption, $c_{i,t}$ has the same support for all $t$, say $[\underline{c_i},\overline{c_i}]$ with probability density function being $f_i$ and cummulative distribution function being $F_i$. Additionally, $b_{i,t}$s become same for all slots and hence we can represent it as $b_i$. (ii) As we would be focusing on a single round, we omit the notation $t$ from the discussion on Optimal Reverse Auction. 

% %%%%%%%%%%%%%%%%%%%%%%%%%%%%%%%%%%%%
% \subsection{Optimal Reverse Auction}
% %%%%%%%%%%%%%%%%%%%%%%%%%%%%%%%%%%%%
% \sg{explain problem formulation here with known $v_i^t$s. you can say, as we focus on querry $x_t$, for this section, we drop $t$ for ease of exposition}

% \sg{put BIC/IIR defns...Myerson's characterization and then optimal auction}

\subsection{Optimal Reverse Auction}
\label{ssec:optiamlauction}
Auctions are typically of two types: (i) forward auction, wherein the auctioneer is a seller, and (ii) reverse auction, where the auctioneer is a buyer.
The seminal work by ~\citet{myerson1981optimal}
showed how to design an optimal forward auction; that is, an auction that maximizes the expected revenue. Researchers also extend these results for different reverse auction setups~(e.g., \citet{iyengar2008procurement}). Here, we show for our settings what an optimal reverse auction is. 

As the bids of other LLM providers are random variables for an arbitrary LLM provider $i$, we work with expected allocations and payments where expectation is w.r.t. the bids of other providers. 

\paragraph{Expected Allocation and Payments} We have,
\[
   \overline{\mathcal{A}_i}({b}_i) 
   \;=\; \mathbb{E}_{b_{-i}} \!\left[\, \mathcal{A}_i({b}_i, b_{-i}) \,\right],
      \qquad
   \overline{p_i}({b}_i) 
   \;=\; \mathbb{E}_{b_{-i}} \!\left[\, p_i({b}_i, b_{-i}) \,\right], 
\]
where the expectation is taken over the bid profile of all other providers, $b_{-i} = (b_1, b_2, \ldots, b_{i-1},\ b_{i+1}, \ldots, b_m)$.
 The expected utility of provider $i$ with true cost $c_i$ is therefore
\begin{equation}
   U_i(b_i|c_i) \;=\; -c_i \cdot \overline{\mathcal{A}_i}(b_i) \,+\, \overline{p_i}(b_i).
\end{equation}

That is, the provider incurs a cost $c_i$ to respond to the query irrespective of its bid $b_i$.

\paragraph{Mechanism Design Objective}
The user seeks to maximize her utility $u_0$ subject to ensuring that providers are incentivized to truthfully reveal their private costs $c_i$. 
This requirement of \emph{incentive compatibility} lies at the heart of 
mechanism design. Additionally, $\mathcal{A},\mathcal{P}$ should ensure that no provider incurs loss. 
%The mechanism must decide on an allocation  rule $\mathcal{}A$ and a payment rule $P$

\begin{definition}[Bayesian Incentive Compatible~\cite{Nisan_Roughgarden_Tardos_Vazirani_2007}]
    An LLM selection mechanism is Bayesian Incentive Compatible (BIC) if for every provider $i \in [M]$,:
    \[U_i(c_i|c_i) \geq U_i(b_i|c_i),~\forall i \in [M]\]
%\[
%\overline{p_i}(c_i) - c_i \cdot \overline{\mathcal{A}_i}(c_i) \ge \overline{p_i}(b_i) - c_i \cdot \overline{\mathcal{A}_i}(b_i)
%\]
\end{definition}

%\sankarshan{i think the last eqn is not needed}

\begin{definition}[Ex Post Individual Rationality (EPIR)~\cite{Nisan_Roughgarden_Tardos_Vazirani_2007}]
    The LLM selection mechanism satisfies Ex-post Individual Rationality (EPIR) if, for every provider $i \in [M]$ and for every possible vector of true costs $c = (c_1, \dots, c_M)$ drawn from the support of the cost distributions, the utility of provider $i$ is non-negative:
\[
u_i(c) = p_i(c) - c_i \cdot \mathcal{A}_i(c) \geq 0, \quad \forall i \in [M], \forall c \in C
\]
where $C=\prod_i \;[\underline{c_i},\overline{c_i}]$.
\end{definition}

The overall problem may thus be viewed as designing an 
\emph{optimal reverse auction} mechanism for allocating queries to providers. It is basically designing $\mathcal{A},\mathcal{P}$ that maximize $\mathfrak{U}$'s expected utility and satisfy the above two properties. 

With these definitions in place, we state Myerson’s classical characterization of 
BIC mechanisms for the reverse auction setting. 
\begin{theorem}[Myerson’s Characterization Theorem for Reverse Auction~\cite{myerson1981optimal}]
    A mechanism $(\mathcal{A}, \mathcal{P})$ is BIC \emph{iff} it satisfies the following two conditions:
\begin{enumerate}
    \item $\overline{\mathcal{A}_i}(\cdot)$ is non-increasing for all $i = 1, \ldots, n.$
    
    \item The utility takes the form:
    \[
    U_i(c_i) = U_i(\overline{c_i}) + \int_{c_i}^{\overline{c_i}} \overline{\mathcal{A}_i}(s)\, ds 
    \quad \forall c_i \in \mathcal{C}_i; \; \forall i = 1, \ldots, n.
    \]
\end{enumerate}
\end{theorem}
Building on this characterization, the optimal reverse auction can be derived. As, these steps are analogous to the original Myerson's design, we directly state the results. We begin by defining \emph{virtual cost} for a provider $i$.
\begin{equation}
\Psi_i(c_i)\;=\;c_i+\frac{F_i(c_i)}{f_i(c_i)},
\end{equation}
and assume a \emph{regularity condition} that $\Psi_i(\cdot)$ is strictly increasing in $c_i$[~\citet{myerson1981optimal}]. Under regularity condition, an optimal reverse auction's allocation rule selects the provider that maximizes the positive virtual surplus $v_i-\Psi_i(c_i)$,
i.e., selects     \[
    i^*\in\arg\max_{i}\{v_i-\Psi_i(c_i)\mid v_i-\Psi_i(c_i)\geq 0\}
    \]
    Here, we assume $\{i:v_i-\Psi_i(c_i)\geq 0\}$ is non-empty. Else, the query cannot be assigned to any provider. 
The payment rule becomes:
\[
p_i(c)=\mathcal{A}_i(c)\,c_i+\int_{c_i}^{\overline{c_i}} \mathcal{A}_i(c_{-i},s_i)\,ds_i
\]
Regularity implies a threshold form for the allocation: for fixed $c_{-i}$
there exists a critical value $z_i(c_{-i})$ such that $\mathcal{A}_i(c_{-i},s_i)=\mathbf{1}\{s_i\le z_i(c_{-i})\}$.
Hence the integral reduces to $\max\{z_i(c_{-i})-c_i,0\}$ and, if $i$ wins, the payment
becomes the critical threshold $p_i(c)=z_i(c_{-i})$. The mechanism is therefore
BIC, individually rational and maximizes $\mathfrak{U}$'s utility. In summary, 

\begin{theorem}[Optimal Reverse Auction]
\label{thm:opt_auc}
Under the regularity assumption that each virtual cost
$\Psi_i(c_i)=c_i+F_i(c_i)/f_i(c_i)$ is strictly increasing in $c_i$, a mechanism where
\begin{enumerate}
  \item $\mathcal{A}$: selects LLM provider $i^*$ such that:  \[i^*\in\arg\max_{i}\{v_i-\Psi_i(c_i)\mid v_i-\Psi_i(c_i)\geq 0\},
    \] 
  \item $\mathcal{P}$: pays the winner
    \[
    p_{i^*}(c)=z_{i^*}(c_{-i^*}),
    z_{i^*}(c_{-i^*})=\sup\{s:\;v_{i^*}-\Psi_{i^*}(s)\ge \max_{j\neq i^*}\{v_j-\Psi_j(c_j)\}\}
    \]
    and pays all others zero,
\end{enumerate}
is an optimal auction for the user to procure LLM service.  \end{theorem}
\if 0
\begin{itemize}
  \item the allocation maximizes the virtual surplus $\sum_i (v_i-\Psi_i(c_i))\mathcal{A}_i(c)$
        subject to feasibility;
  \item the allocation rule is monotone in each provider's reported cost;
  \item the induced payment rule is incentive compatible (ex-post truthful) and individually rational;
  \item the winner's payment equals the critical threshold $z_{i^*}(c_{-i^*})$, i.e.,
        the Myerson form
        $p_i(c)=\mathcal{A}_i(c)c_i+\int_{c_i}^{\overline{c_i}} \mathcal{A}_i(c_{-i},s)\,ds$
        reduces to $p_{i^*}=z_{i^*}(c_{-i^*})$ when $i^*$ wins.
\end{itemize}

\fi 
\paragraph{Different $n_{i,q_t}$ for different Queries} Note that, if $n_{i,q_t}$ are different for each query, which is very natural, $\Psi_i$'s would depend upon query. In that case, we would have to replace $\Psi_i$ as $\Psi_{i,t}$ which in turn depend upon $n_{i,q_t}$. The overall notation would become clumsy and hence, we prefer to skip these technicalities for ease of exposition. 
Furthermore, to introduce an optimal reverse auction, we assumed $v_{i,t}=\mathbb{E}[r_{i,t}|x_t] = \theta_i^{\top}x_t$ to be known. However, in our LLM provider-user setup, $\theta_i$s are unknown but can be learnt.
%%%%%%%%%%%%%%%%%%%%%%%%%%%%%%%%%%%%
\subsection{Stochastic and Unknown Rewards}
%%%%%%%%%%%%%%%%%%%%%%%%%%%%%%%%%%%%
\paragraph{History and Learning}
In practical settings, the reward obtained from selecting a provider is stochastic. The user $\mathfrak{U}$ is interested in its mean $v_{i,t} \triangleq \mathbb{E}[r_{i,t}|x_t]=\theta_i^{\top}x_t$ where $\theta_i \in \mathbb{R}^d$ is unknown to $\mathfrak{U}$ and $x_t \in \mathbb{R}^d$ captures the context of the query at round $t$. Let $r_{I_t,t}$ denote the realized reward from provider $I_t \in [M]$ at round $t \in [T]$. Here $I_t \in [M]$ denotes the provider selected for query $x_t$ in round $t$. $\mathfrak{U}$ can estimate (learn) $\theta_i$ from historical observations.  Let $h_t$ capture the history till round $t$.
\[
   h_t \;=\; h_{t-1} \cup \{ x_t, I_t, r_{I_t,t} \}, 
   \qquad h_0 = \emptyset,
\]
We model the interaction as a \emph{contextual multi-armed bandit} (MAB) problem~\cite{lu2019conmab}, in particular, the above settings is also referred to as linear contextual MAB. Let $\mathcal{ALG}$ be the algorithm that learns $\theta_i$s and assigns a new query $x_t$ to one of the providers based on $x_t \mbox{ and } h_{t-1}$. 

The quality of a contextual bandit algorithm $\mathcal{ALG}$ is typically measured by its \emph{regret}, defined as the difference between the cumulative expected reward of an oracle that always selects the optimal provider (given full knowledge of $\theta_i$'s), and that of the learning algorithm. Formally, if $i^*_t \in \arg\max_i \theta_i^\top x^t$ denotes the optimal provider at round $t$, the cumulative regret up to horizon $T$ is
\begin{equation}
\mathbb{R}_T(\mathcal{ALG}) = \sum_{t=1}^{T} \left[ \left(\theta_{i_t^*}^\top x_t - \psi_{i_t^*} \left(b_{i_t^*}\right) \right) - \left(\theta_{I_t}^\top x_t - \psi_{I_t} \left(b_{I_t}\right) \right) \right],
\end{equation}
where $I_t$ is the provider chosen by the learning algorithm $\mathcal{ALG}$ at round $t$. The objective is to design allocation mechanisms that achieve sublinear regret (lower the better), i.e., $R(T) = o(T)$, ensuring asymptotic convergence to the optimal allocation.
With appropriately defining rewards, one can easily adapt LinUCB~\cite{li2010contextual} or SupLinUCB~\cite{chu2011contextual} for our settings. LinUCB-based algorithm would incur $\Omega(T)$ regret, though practically can perform better. SUP-LinUCB-based algorithm would incur $O(\sqrt{T})$ regret. The challenge in the SUP-LinUCB is, when the costs are private, the providers can manipulate the learning algorithm. Thus, there is need to carefully design MAB algorithm and payments~\cite{devanur2009mmab,babaioff2009characterizing}. Such auction design is called \emph{MAB Mechanism Design}.

%%%%%%%%%%%%%%%%%%%%%%%%%%%%%%%%%%%%
% \subsection{MAB Mechanism Design}
%%%%%%%%%%%%%%%%%%%%%%%%%%%%%%%%%%%%

\paragraph{MAB Mechanism Design} \citet{babaioff2009characterizing}, ~\citet{devanur2009mmab} address this challenge by first showing that any deterministic truthful MAB mechanism must be exploration-separated, and inevitably suffer a regret in the order of $\Omega(T^{2/3})$. That is, there exists a fundamental trade-off between truthfulness and learning efficiency. To overcome this limitation, \citet{babaioff2015truthful} propose a randomization via random resampling of bids, and demonstrate that it is possible to retain regret guarantees of a \emph{monotone} learning algorithm in the presence of the strategic providers. However, their resampling procedure is designed for forward auctions and aims to optimize social welfare. Contrarily, our goal is to optimize the user's utility, i.e., design an optimal reverse auction. 

% \paragraph{Our Approach} To tackle user's problem, we need a monotone allocation rule for contextual MAB with appropriately defined reward function and a resampling procedure. We elaborate on this in the next section. 
%%%%%%%%%%%%%%%%%%%%%%%%%%%%%%%%%%%%%%%%%%%%%%%%%%%%%%%%%%%%%%%%%%%%%%%%

\section{Our Approach}
\label{sec:our_apprch}
%%%%%%%%%%%%%%%%%%%%%%%%%%%%%%%%%%%%%%%%%%%%%%%%%%%%%%%%%%%%%%%%%%%%%%%%
% \sankarshan{pick either "resampling" or "resampling" and stick with it}\sg{lets stick with resampling}

In a nutshell, to design an optimal reverse auction with learnable rewards, we require a monotone allocation rule for contextual MAB with appropriately defined reward function and a resampling procedure. Towards this, we begin by defining monotonicity of the allocation rule. Then, in Section~\ref{ssec:montonealloc}, we propose two algorithms: \ourlinucb\ and \oursupucb\ adapted from SupLinCUB~\cite{chu2011contextual} to ensure monotonicity. Section~\ref{ssec:resample} introduces our resampling procedure, \ourrealgo, and proves that any monotone allocation rule through a contextual MAB algorithm ($\mathcal{ALG}$) combined with \ourrealgo\ ensures truthful reporting of the costs by the providers while retaining regret guarantees of $\mathcal{ALG}$. 

\subsection{Ex-Post Monotonicity}

The ex-post monotonicity condition ensures that the allocation rule respects the economic intuition that bidding more competitively (i.e., lowering one’s declared cost) should not disadvantage a provider in terms of allocation opportunities. More formally,

\begin{definition}[Ex Post Monotone]
An allocation rule \( \mathcal{A} \) is \emph{ex-post monotone} for a reverse auction if, for every possible sequence of context arrivals ($x_1,x_2,\ldots,x_T$) and reward realizations, for each provider \( i \in [M] \), for all bids \( b_{-i} \) of other providers, and for any two possible bids of provider \( i \), \( b_i' \leq b_i \), we have:
\[
\mathcal{A}_i(b_i', t) \geq \mathcal{A}_i(b_i, t) \quad \text{for all rounds } t.
\]
Here, \( \mathcal{A}_i(b,t) \) denotes the total number of times provider \( i \) is allocated in the first \( t \) rounds when the bid vector is \( b \).
\end{definition}

This property is  in establishing the truthfulness of mechanisms within our stochastic reverse auction setting. Let $\mathcal{ALG}$ be a monotone allocation rule.

%%%%%%%%%%%%%%%%%%%%%%%%%%%%%%%%%%%%
\subsection{Monotone MAB Algorithms}
\label{ssec:montonealloc}
%%%%%%%%%%%%%%%%%%%%%%%%%%%%%%%%%%%%
In the literature, two of the most popular MAB algorithms for contextual bandits are, LinUCB~\cite{li2010contextual} and SupLinUCB~\cite{chu2011contextual}. We can leverage these algorithms using net rewards as, $v_{i, t} - \Psi_i(b_i)$; $\Psi_i(b_i)$ is subtracted as $\mathfrak{U}$'s goal is to deploy optimal reverse auction, which selects a provider having the highest $v_{i,t}-\Psi_i(b_i)$. However, with just this update, the allocation won't be monotone. To this extent, we propose \ourlinucb\ (Algorithm~\ref{alg:baselinucb}) and \oursupucb\ (Algorithm~\ref{alg:suplinucb}) that provide monotone allocations.

\paragraph{Difference between \oursupucb\ and SupLinUCB} 
 Broadly, \oursupucb\  works similarly to SupLinUCB. I.e., it works in stages. Each stage consists of twice the previous number of rounds, thus a total $O(\log T)$ stages. In each stage, the algorithm maintains an active set of providers. As the stage progresses, the active set keeps eliminating the suboptimal providers. However, unlike SupLinUCB, which considers a common $\theta$, that is $\theta_i$s are same, the context for each provider is different. As such, \oursupucb\ has independent $\theta_i$s to be learned for each provider $i$. 

To make SupLinUCB monotone, we decouple learning from the auction by limiting the learning only by the round robin method. Compared to SupLinUCB, we select a provider from the active set one by one (refer to Line 7 of Algorithm~\ref{alg:suplinucb}). Also, we add an exploitation condition, Lines 25-27, wherein we exploit within the current stage. We consider the reward with the confidence subtracted by the virtual cost to match it with the optimal reverse auction setting (as described in Section~\ref{ssec:optiamlauction}).
I.e., the selection is done using   $\hat{v}_{i, t} + w_{i,t}^s - \Psi_i(b_i)$; $w_{i,t}^s$ is the confidence term in UCB.  We defer the proof of the monotonicity of \oursupucb\ to Section~\ref{sec:thy_ana}.

\begin{algorithm}[t]
\caption{\ourlinucb}
\label{alg:baselinucb}
\begin{algorithmic}[1]
\State \textbf{Input:} Parameter $\alpha \in \mathbb{R}^+$, History sets $\{H_{i,t}\}_{i \in [M]}$, where $\Lambda_{i,t} \subseteq \{1, 2, \dots, t-1\}$
\State Observe context vector $x_t \in \mathbb{R}^d$
\For{each provider $i \in [M]$}
    \State $A_{i,t} \gets I_d + \sum_{\tau \in H_{i,t}} x_\tau x_\tau^T$
    \State $g_{i,t} \gets \sum_{\tau \in H_{i,t}} r_{i,\tau} x_\tau$
    \State $\hat{\theta}_{i,t} \gets A_{i,t}^{-1} g_{i,t}$
    \State \textsl{// Calculate value estimate and confidence width}
    \State $\hat{v}_{i,t}^s \gets \hat{\theta}_{i,t}^\top x_t$
    \State $w_{i,t}^s \gets \alpha \sqrt{x_t^T A_{i,t}^{-1} x_t}$
\EndFor
\State \textbf{Return:} $\{ \hat{v}_{i,t}^s \}_{i \in [M]}$ and $\{ w_{i,t}^s \}_{i \in [M]}$
\end{algorithmic}
\end{algorithm}
%%%%%%%%%%%%%
\begin{algorithm}[h]
\caption{REV-SupLinUCB-S-OPT}
\label{alg:suplinucb}
% \small
\begin{algorithmic}[1]
\State \textbf{Input:} Virtual costs $\Psi(b) = (\Psi_1(b_1), \dots, \Psi_n(b_n))$
\State \textbf{Initialization:}
\State $S_{\text{max}} \gets \lceil \ln{T} \rceil$
\State $\Lambda_{i,t}^s \gets \emptyset$ for all $i \in [M], t \in [T], s \in [S_{\text{max}}]$
\For{$t = 1, 2, \dots, T$}
    \State $s \gets 1$ and $\hat{A}_1 \gets [M]$
    \State $j \gets 1 + (t \mod n)$
    \Repeat
        \State \parbox[t]{.85\linewidth}{Use REV-BaseLinUCB-S with index sets $\{ \Lambda_{i,t}^s \}_{i \in [M]}$ and context $x_t$ to get $(\hat{v}_{i,t}^s)_{i \in \hat{A}_s}$ and $(w_{i,t}^s)_{i \in \hat{A}_s}$.}
        \If{$j \in \hat{A}_s$ and $w_{j,t}^s > 2^{-s}$} 
            \State \textsl{// Forced exploration}
            \State Select $I_t = j$
            \State \parbox[t]{.85\linewidth}{Update history sets for all stages $s' \in [S_{\text{max}}]$: \\ $\Lambda_{i,t+1}^{s'} \leftarrow 
\begin{cases} 
\Lambda_{i,t}^{s'} \cup \{t\} & \text{if } s = s' \\
\Lambda_{i,t}^{s'} & \text{otherwise}
\end{cases}$
}
        \ElsIf{$w_{i,t}^s \leq \frac{1}{\sqrt{T}}$ $\forall \ i \in \hat{A}_s$} 
            \State \textsl{// Pure exploitation (final phase)}
            \State Select $I_t = \arg\max_{i \in \hat{A}_s} \left\{ (\hat{v}_{i,t}^s + w_{i,t}^s) - \Psi_i(b_i) \right\}$
            \State \parbox[t]{.85\linewidth}{Update index sets for $I_t$ at all stages: \\ $\Lambda_{I_t,t+1}^{s'} \gets \Lambda_{I_t,t}^{s'} \quad\forall s' \in [S]$}
            \State $\Lambda_0 \gets \Lambda_0 \,\cup \{t\}$
        \ElsIf{$w_{i,t}^s \leq 2^{-s}$ $\forall \ i \in \hat{A}_s$} 
            \State // Stage advancement
            \State $\hat{A}_{s+1} \gets \big\{ i \in \hat{A}_s \mid$
            \State \quad $(\hat{v}_{i,t}^s + w_{i,t}^s) - \Psi_i(b_i) \geq$
            \State \quad $\displaystyle \max_{a \in \hat{A}_s} \big\{ (\hat{v}_{a,t}^s + w_{a,t}^s) - \Psi_a(b_a) \big\} - 2 \cdot 2^{1-s} \big\}$
            \State $s \gets s + 1$
        \Else
            \State \textsl{// Exploitation within the current stage}
            \State Select $I_t = \arg\max_{i \in \hat{A}_s} \left\{ (\hat{v}_{i,t}^s + w_{i,t}^s) - \Psi_i(b_i) \right\}$
            % \State \parbox[t]{.85\linewidth}{Update history sets for $I_t$ at all stages: $\Lambda_{I_t,t+1}^{s'} \gets \Lambda_{I_t,t}^{s'} \cup \{t\}$.}
            \State $\Lambda_{est}^s \gets \Lambda_{est}^s \cup \{t\}$
        \EndIf
    \Until{$I_t$ is selected}
\EndFor
\end{algorithmic}
\end{algorithm}

%%%%%%%%%%%%%%%%%%%%%%%%%%%%%%%%%%%%
\subsection{\ourrealgo: Resampling Procedure}
\label{ssec:resample}
%%%%%%%%%%%%%%%%%%%%%%%%%%%%%%%%%%%%

To avoid exploration-separated mechanisms that incur higher regret for truthful implementation, we need randomization in the mechanism. Towards this, we randomly update the bids of the providers through a \emph{resampling procedure}. Motivated from \cite[Definition 4.3]{babaioff2015truthful}, we formally define it for our setting next.

\begin{definition}[Reverse Self-Resampling Procedure]
\label{def:self_resample}
Let $I$ be a nonempty interval in $\mathbb{R}$.  
A \emph{reverse self-resampling procedure} with support $I$ and resampling probability $\mu \in (0, 1)$ is a randomized algorithm with input $b_i \in \mathcal{I}$, random seed $rs_i$, and output $\widetilde{b_i}(b_i; rs_i) \in I$, that satisfies the following properties:

\begin{enumerate}
    \item For every fixed $rs_i$, $\widetilde{b_i}(b_i; rs_i)$ is non-decreasing in $b_i$.
    \item With probability $1 - \mu$, we have $\widetilde{b_i}(b_i; rs_i) = b_i$. Otherwise, $\widetilde{b_i}(b_i; rs_i) > b_i$.
    \item Consider the two-variable function
    \[
    F(a_i, b_i) = \Pr[\, \widetilde{b_i}(b_i; rs_i) < a_i \mid \widetilde{b_i}(b_i; rs_i) > b_i \,], \forall \ a_i > b_i 
    \]
    which is called as the \emph{distribution function} of the reverse self-resampling procedure.  
    For each $b_i$, the function $F(\cdot, b_i)$ must be differentiable and strictly increasing on the interval $\mathcal{I} \cap (b_i, \overline{c_i})$.
\end{enumerate}
\end{definition}

Algorithm~\ref{alg:rosa} presents a construction for a reverse auction with the self-sampling procedure from Definition~\ref{def:self_resample}.

%%%%%%%%%%%%%%%%%%%%%%%%%%
\begin{algorithm}[h]
\caption{\ourrealgo: Reverse OneShot Adjustment, A Non-Recursive Self-Resampling Procedure for Reverse Auction}
\label{alg:rosa}
\begin{algorithmic}[1]
\Require Bid $b_i \in [0, \infty)$, resampling parameter $\mu \in [0,1]$, upper cost bound $\overline{c_i}$, random seed $rs_i$
\State Sample a random variable $\gamma_i$ uniformly from $[0,1]$
\State Set resampling factor \[\xi_i = 
\begin{cases}
1 & \text{with probability } 1 - \mu \\
1 + \gamma_i \left( \dfrac{\overline{c_i}}{b_i} - 1 \right) & \text{with probability } \mu
\end{cases}\]
\State Construct and return the modified bid $\widetilde{b_i}(b_i, rs_i) = \xi_i b_i$
\end{algorithmic}
\end{algorithm}
%%%%%%%%%%%%%%%%%%%%%%%%%%

\begin{proposition}[]
\ourrealgo\  is a reverse self-resampling procedure with support $\mathbb{R}^+$ and resampling probability $\mu$.  
The distribution function for this procedure is $F(a_i, b_i) = \frac{a_i - b_i}{\overline{c_i} - b_i}$.
\end{proposition}

\begin{proof}
Properties~1 and~2 in Definition~\ref{def:self_resample} for \ourrealgo\ follow immediately from the description of Algorithm~\ref{alg:rosa}.  For Property~3, by conditioning on the event $\widetilde{b_i}(b_i; rs_i) > b_i$, we see that the distribution of $\widetilde{b_i}(b_i; rs_i)$ is uniform in $[b_i, a_i]$.  
\end{proof}

Note that, typically, for each provider we need to use different random seed $rs_i$. Hence, for each provider, the resampling procedure is treated as $rp_i$. Let $\mathbf{RP}$ be ensemble of these procedure.  

% --------------------- Algorithm ---------------------

We now illustrate that with a monotone MAB algorithm $\mathcal{ALG}$, one can design a truthful reverse auction using \ourrealgo.

% --------------------- Mechanism ---------------------
\begin{figure}[t]
\centering
\begin{tcolorbox}[mechbox, title={\textsc{Rev-GTM}: Generic Reverse Transformation Mechanism}]
\label{mech:revgtm}
\begin{algorithmic}[1]
\State \textbf{Input:} Bid vector $b = (b_1, b_2, \dots, b_M)$
\State \textbf{Output:} Allocation and payment $(\mathcal{A}, \mathcal{P})$
\For{each provider $i \in [M]$}
    \State Obtain $\widetilde{b_i} = \ourrealgo(b_i, \mu, \overline{c_i}, rs_i)$
\EndFor
\State Construct modified bid vector $\widetilde{b} = (\widetilde{b_1}, \widetilde{b_2}, \dots, \widetilde{b_M})$
\State \textbf{Allocation:} Allocate according to monotone rule $\mathcal{ALG}(\widetilde{b})$
\State \textbf{Payment:}
\For{each provider $i \in [M]$}
    \If{$\xi_i = 1$}
        \State Set $p_i \gets b_i \cdot \mathcal{A}_i(x)$
    \Else
        \State Set $p_i \gets b_i \cdot \mathcal{A}_i(x) + \dfrac{1}{\mu} \cdot \mathcal{A}_i(x) \cdot (\overline{c_i} - b_i)$
    \EndIf
\EndFor
\end{algorithmic}
\end{tcolorbox}
\end{figure}

Suppose we are given a monotone allocation rule $\mathcal{ALG}$ and a vector $\mathbf{RP}$ of reverse self-resampling procedures of each provider $i \in [M]$ that has a resmapling probability $\mu \in (0, 1)$, support $\mathcal{T}_i$, and output value $\widetilde{b_i}(b_i; rs_i)$. 
Based on the resampling procedure, we propose a generic transformation that combines $\mathcal{ALG}, \mu, \mathbf{RP}$ into a randomized mechanism $\mathcal{M}={\textsc{Rev-GTM}}(\mathcal{ALG}, \mu, \mathbf{RP})$. 
Next, we introduce a couple of desirable properties of any randomized mechanism. 

\begin{definition}[Truthfulness in Expectation (Ex-post Incentive Compatibility, EPIC)]
    For every provider $i\in [M]$ and for every realization of other provider' bids $b_{-i}$, truthful reporting of its cost $b_i = c_i$ maximizes the provider’s expected utility over the mechanism’s internal randomness. That is, $\forall i\in M,\forall t$
    $$u_{i,t}(c_i,b_{-i}|c_i)\geq u_{i,t}(b_i,b_{-i}|c_i) \forall b_i, \forall b_{-i},\forall h_{t-1} $$

\end{definition}

\begin{definition}[Universal Ex-post Individual Rationality (EPIR)]
    For every realization of the mechanism’s randomness, each provider’s realized utility is non-negative when bidding truthfully, i.e., $\forall i\in [M],\forall t, \forall h_{t-1}$
    \[
    u_{i,t}(c_i, b_{-i}|c_i) \ge 0.
    \]
\end{definition}

Now, we state important result of this section. 

\begin{theorem}
\label{thm:truthful_epic_epir}
Consider an arbitrary single-parameter reverse auction domain, and let $\mathcal{ALG}$ be a monotone allocation rule. Suppose we are given an ensemble $\mathbf{RP}$ of self-resampling procedures, where each procedure has resampling probability $\mu \in (0,1)$.
Let the mechanism $\mathcal{M} = ({\mathcal{A}}, {\mathcal{P}}) = {\textsc{Rev-GTM}}(\mathcal{ALG}, \mu, \mathbf{RP})$ be the mechanism constructed from $\mathcal{ALG}$ via the self-resampling transformation.
Then the mechanism $\mathcal{M}$ satisfies the following properties:
\begin{enumerate}
    \item EPIC, EPIR
    \item For $m$ providers and any bid vector $b$ (and any fixed random seed of nature), allocations $\mathcal{A}(b)$ and $\mathcal{ALG}(b)$ are identical with probability at least $1 - M\mu$.

    \item If $T = \mathbb{R}_{+}^{M}$ (all types are positive), and each $\mathbf{RP}_i$ is the reverse self-resampling procedure, then mechanism $\mathcal{M}$ is ex-post no-positive-transfers, and never pays any provider $i$ more than 
    \[
        b_i \cdot \mathcal{ALG}_i(\widetilde{b}) + (\overline{c_i} - b_i) \cdot \mathcal{ALG}_i(\widetilde{b})\left( \frac{1}{\mu} \right).
    \]
\end{enumerate}
\end{theorem}

\begin{proof}
We begin by stating a theorem which helps to estimate integrals via random sampling. 
\begin{theorem}[Theorem 4.2 from \cite{babaioff2015truthful}]
\label{thm:int_est}
Let $\mathcal{I}$ be an non-empty open interval in $\mathbb{R}$ and $F : \mathcal{I} \to [0, 1]$ be any strictly increasing function that is differentiable and satisfies 
\[
\inf_{z \in \mathcal{I}} F(z) = 0 \quad \text{and} \quad \sup_{z \in \mathcal{I}} F(z) = 1.
\]
If $Y$ is a random variable with cumulative distribution function $F$, then
\[
\int_\mathcal{I} g(z) \, dz = \mathbb{E} \!\left[ \frac{g(Y)}{F'(Y)} \right].
\]
\end{theorem}

\textbf{Part (1)}
The argument follows the structure of~\cite{babaioff2015truthful}.  
To establish truthfulness, it suffices to show:  
(i) the allocation rule $\mathcal{A}$ is monotone, and  
(ii) the payment rule satisfies Myerson's characterization.  

Monotonicity of $\mathcal{A}$ follows from the monotonicity of $\mathcal{ALG}$ and Property~1 of the reverse self-resampling procedure, which ensures that increasing a bid cannot reduce the probability of allocation.  

Note when we write $\mathcal{A}$, we consider the expected allocation over the $rs$ used to alter the bids. 

We have from Myerson's Characterization. 
\begin{align*}
    \tilde{p}_i(b) &= c_i \mathcal{A}_i(b_i) + \int_{c_i}^{\overline{c_i}}\mathcal{A}_i(b_{-i}, u)du 
\end{align*}

The payment in the mechanism $\mathcal{M}$ is:
\begin{align*}
    p_i(b) &= b_i \mathcal{ALG}_i(\widetilde{b}) + R_i \text{, where} \\
    R_i &= \begin{cases} 
\frac{1}{\mu} \mathcal{ALG}_i(\widetilde{b}) \left( \overline{c_i} - c_i \right) & \text{if } \xi > 1 \\
0 & \text{if } \xi = 1
\end{cases}
\end{align*}

We have to show that the expected payment of our mechanism matches Myerson's Characterization. 

\[
\mathbb{E}[p_i(b)] = \mathbb{E}[b_i\mathcal{ALG}_i(\widetilde{b})] + \mathbb{E}[R_i]
% = b_i\mathcal{A}_i(b) + \mathbb{E}[R_i].
\]
We have:
\[
\mathbb{E}[b_i \mathcal{ALG}_i(\widetilde{b})] = b_i \mathcal{A}_i(b)
\]

We need to show that, 
\begin{align*}
    \mathbb{E}[R_i] &= \int_{c_i}^{\overline{c_i}} \mathcal{A}_i(b_{-i}, u) \, du
\end{align*}

Using Theorem~\ref{thm:int_est} with random variable $X = \widetilde{b_i}$ and function $g(u) = \mathcal{A}_i(b_{-i}, u)$, we get
\begin{align*}
    \int_{b_i}^{\overline{c_i}} \mathcal{A}_i(b_{-i}, u) \, du &= \mathbb{E}\left[ \frac{\mathcal{A}_i(b_{-i}, \widetilde{b_i}(b_i; rs_i))}{F'(\widetilde{b_i}(b_i; rs_i), b_i)} \ \middle| \ \overline{c_i} > \widetilde{b_i}>b_i\right] \\
    &= \mathbb{E}\left[ \frac{\mathcal{ALG}_i(\widetilde{b})}{F'(\widetilde{b_i}, b_i)} \ \middle| \ \overline{c_i} > \widetilde{b_i}>b_i\right] \\
    &= \mathbb{E}\left[ (\overline{c_i} - b_i)\mathcal{ALG}_i(x)\ \middle| \ \overline{c_i} > \widetilde{b_i}>b_i\right] \\
    &= \mu \cdot \mathbb{E}\left[ R_i\ \middle| \ \overline{c_i} > \widetilde{b_i} > b_i\right] \\
    &= \mu \cdot \mathbb{E}\left[ R_i\right] \\
\end{align*}

The last line follows from the fact that $R_i$ is non-zero only when the bid is resampled. 

\textbf{Part (2)}
We can observe from the algorithm that the bids are resampled only with probability $\mu$. So, the allocations won't alter if none of the bids are altered. And the probability of this event would be $1-M \mu$

\textbf{Part(3)}
One can observe from the payment structure that the reverse self-resampling procedure satisfies ex-post no-positive-transfer. 
\end{proof}

% MERGE THE PRROF SKETCH AND PROOF

% \begin{proofsketch}
% \textbf{(1)} The argument follows the structure of~\cite{babaioff2015truthful}.  
% It suffices to show:  
% (i) the allocation rule $\mathcal{A}$ is monotone, and  
% (ii) the payment rule satisfies Myerson's characterization.  

% Monotonicity of $\mathcal{A}$ follows from the monotonicity of $\mathcal{ALG}$ and Property~1 of the reverse self-resampling procedure, which ensures that increasing a bid cannot reduce the probability of allocation.  

% For the payment, the expected payment of provider $i$ is  
% \[
% \mathbb{E}[p_i(b)] = \mathbb{E}[b_i\mathcal{ALG}_i(\widetilde{b})] + \mathbb{E}[R_i]
% = b_i\mathcal{A}_i(b) + \mathbb{E}[R_i].
% \]
% Hence, to match Myerson’s characterization, it remains to show
% \[
% \mathbb{E}[R_i] = \int_{c_i}^{\overline{c_i}} \mathcal{A}_i(b_{-i},u)\,du.
% \]
% This equality holds because the random resampling process, unbiasedly estimates the integral term in Myerson’s formula.  

% Individual rationality follows directly since the expected utility of each provider is non-negative under truthful reporting and monotone allocation.  

% \textbf{(2)} Since each bid is resampled independently with probability $\mu$, the probability that no bids are resampled equals $1 - M\mu$.  

% \textbf{(3)} Finally, ex-post no-positive-transfer holds directly from the payment rule.
% \end{proofsketch}

%%%%%%%%%%%%%%%%%%%%%%%%%%%%%%%%%%%%%%%%%%%%%%%%%%%%%%%%%%%%%%%%%%%%%%%

\subsection*{\ourmabalgo}

We propose that the user $\mathfrak{U}$  deploys the auction as follows: \textit{First}, it collects the bids from all the LLM providers and updates them through \ourrealgo. \textit{Second}, with updated bids, for every new query, it selects the provider recommended by \oursupucb\ and the payments are computed as given \textsc{Rev-GTM}. We refer to such an auction as \ourmabalgo.
In summary, \ourmabalgo = {\textsc{Rev-GTM}}(\oursupucb, $\mu$, $\mathbf{RP}$).

\section{\ourmabalgo: Theoretical Analysis}
\label{sec:thy_ana}
%%%%%%%%%%%%%%%%%%%%%%%%%%%%%%%%%%%%%%%%%%%%%%%%%%%%%%%%%%%%%%%%%%%%%%%%
We now prove that \oursupucb\ is monotone and hence from Theorem~\ref{thm:truthful_epic_epir}, \ourmabalgo\ is truthful reverse contextual MAB auction that optimizes the user's utility.

\subsection{Monotonicity of \ourmabalgo}

\begin{theorem}
\label{thm:mono}
    The allocation rule induced by \oursupucb\ is ex-post monotone.
\end{theorem}

\begin{proof}
We consider two bid vectors, 
\[
b = (b_i, b_{-i}) \quad \text{and} \quad b^- = (b_i^-, b_{-i}),
\]
where \(b_i^- < b_i\) and all other bids \(b_{-i}\) are fixed. The analysis assumes a fixed, arbitrary sequence of contexts \(\{x_t\}_{t=1}^T\) and their corresponding reward realizations.

\begin{claim}[Bid-Independence of Learning]
The value estimates \(\hat{v}_{j,t}^k\) and confidence widths \(w_{j,t}^k\) computed by the \oursupucb\  for any provider \(j\) at any round \(t\) and stage \(k\) are independent of the bid vector.
\end{claim}

\begin{proof}
\oursupucb\ computes its outputs using only the index sets \(\Lambda_{j,t}^s\), the sequence of contexts \(x_\tau\), and rewards \(r_{j,\tau}\) for \(\tau \in \Lambda_{j,t}^s\). In the main REV-SupLinUCB-S-OPT algorithm, the index sets are updated (i.e., learning occurs) only when a provider $i$ is selected during the forced exploration phase (Line [10]) or after an exploitation choice (Line [14]). The forced exploration selection is based on a bid-independent round-robin schedule \(j \gets 1 + (t \bmod n)\). 

\end{proof}

\begin{claim}[Monotonicity of Active Sets]
For any two bid vectors \(b\) and \(b^-\) as defined above, for all rounds \(t \in \{1, \ldots, T\}\) and all stages \(k \in \{1, \ldots, K_{\max}\}\), the active sets satisfy the subset property:
\[
\hat{A}_k(b, t) \subseteq \hat{A}_k(b^-, t).
\]
This claim states that a provider is always active under a more competitive (lower) bid if it is active under a less competitive (higher) bid.
\end{claim}

\begin{proof}[Proof by Induction]

\textbf{Base Case (\(t = 1\)):} At the beginning of round \(t = 1\), the algorithm initializes. 

For the first stage:
\[
\hat{A}_1(b, 1) = M \quad \text{and} \quad \hat{A}_1(b^-, 1) = M.
\]
Thus, \( \hat{A}_1(b, 1) \subseteq \hat{A}_1(b^-, 1) \) holds trivially.

For subsequent stages \(k > 1\) at \(t = 1\), the active sets are determined by elimination steps, following the same logic as the inductive step. Hence, the base case holds.

\vspace{1em}
\textbf{Inductive Hypothesis:} Assume the claim holds for all rounds up to \(\tau - 1\), i.e.,
\[
\forall\, t' < \tau,\, \forall\, k: \quad \hat{A}_k(b, t') \subseteq \hat{A}_k(b^-, t').
\]

\vspace{1em}
\textbf{Inductive Step: Prove for round \(t = \tau\):} We prove by nested induction on stage \(k\) that
\[
\forall\, k: \quad \hat{A}_k(b, \tau) \subseteq \hat{A}_k(b^-, \tau).
\]

\vspace{1em}
\underline{Inner Base Case (\(k = 1\)):} At the start of round \(\tau\),
\[
\hat{A}_1(b, \tau) = \mathbb{N}, \quad \hat{A}_1(b^-, \tau) = \mathbb{N}.
\]

\vspace{1em}
\underline{Inner Inductive Hypothesis:} Assume \( \hat{A}_k(b, \tau) \subseteq \hat{A}_k(b^-, \tau) \) for some stage \(k\).

\vspace{1em}
\underline{Inner Inductive Step (Prove for stage \(k+1\)):}

Let \(j \in \hat{A}_{k+1}(b, \tau)\). By definition of the elimination rule (Line 16), this implies:

1. \(j \in \hat{A}_k(b, \tau)\)

2. \(j\) satisfies the survival condition:
\[
(\hat{v}_{j,\tau}^k + w_{j,\tau}^k) - \Psi_j(b_j) \ge \max_{a \in \hat{A}_k(b, \tau)} \left\{ (\hat{v}_{a,\tau}^k + w_{a,\tau}^k) - \Psi_a(b_a) \right\} - 2^{1-k}.
\]

From the inner inductive hypothesis, \(j \in \hat{A}_k(b^-, \tau)\) as well. Now we show that \(j\) satisfies the survival condition under \(b^-\).

\vspace{1em}
\textbf{Learning Independence:} The history sets \(\Lambda_{j,\tau}^k\) are identical under \(b\) and \(b^-\) due to the bid-independent forced exploration rule. Thus,
\[
\hat{v}_{j,\tau}^k \text{ and } w_{j,\tau}^k \text{ are identical under both } b \text{ and } b^-.
\]

\vspace{1em}
\textbf{Virtual Cost Regularity:}
\begin{itemize}
    \item For all \(j \ne i\): \(\Psi_j(b_j^-) = \Psi_j(b_j)\)
    \item For \(j = i\): since \(b_i^- < b_i\), we have \( \Psi_i(b_i^-) < \Psi_i(b_i) \)
\end{itemize}

Define:
\[
\mathrm{OVS}_a(b_a) = (\hat{v}_{a,\tau}^k + w_{a,\tau}^k) - \Psi_a(b_a).
\]

The survival condition becomes:
\[
\mathrm{OVS}_j(b_j^-) \ge \max_{a \in \hat{A}_k(b^-, \tau)} \mathrm{OVS}_a(b_a^-) - 2^{1-k}.
\]

We compare the two sides:

\textbf{LHS:} Since \(\Psi_j(b_j^-) \le \Psi_j(b_j)\), we have
\[
\mathrm{OVS}_j(b_j^-) \ge \mathrm{OVS}_j(b_j).
\]

\textbf{RHS:} The set \( \hat{A}_k(b, \tau) \subseteq \hat{A}_k(b^-, \tau) \) and \( \mathrm{OVS}_a(b_a^-) \ge \mathrm{OVS}_a(b_a) \) for all \(a\). Therefore,
\[
\max_{a \in \hat{A}_k(b^-, \tau)} \mathrm{OVS}_a(b_a^-) \ge \max_{a \in \hat{A}_k(b, \tau)} \mathrm{OVS}_a(b_a).
\]

Hence,
\begin{align*}
\mathrm{OVS}_j(b_j^-) &\ge \mathrm{OVS}_j(b_j) \\
&\ge \max_{a \in \hat{A}_k(b, \tau)} \mathrm{OVS}_a(b_a) - 2^{1-k} \\
&\Rightarrow \mathrm{OVS}_j(b_j^-) \ge \max_{a \in \hat{A}_k(b^-, \tau)} \mathrm{OVS}_a(b_a^-) - 2^{1-k}.
\end{align*}

Thus, \(j \in \hat{A}_{k+1}(b^-, \tau)\), proving the inner inductive step.

\vspace{1em}
By induction on both \(t\) and stage \(k\), we conclude that
\[
\forall\, t, \forall\, k: \quad \hat{A}_k(b, t) \subseteq \hat{A}_k(b^-, t).
\]
\end{proof}

\begin{claim}[Monotonicity of Selection for provider \(i\)]
In any round \(t\) and stage \(k\) where an exploitation choice is made, provider \(i\)'s optimistic virtual surplus is strictly higher with bid \(b_i^-\) than with bid \(b_i\), making it a more attractive candidate for selection.
\end{claim}

\begin{proof}
The selection rule during exploitation is
\[
\arg\max_{a \in \hat{A}_k} \mathrm{OVS}_a(b_a).
\]
From Claim 1, the terms \(\hat{v}_{i,t}^k\) and \(w_{i,t}^k\) are independent of the bid. From Regularity condition, we have, \(\Psi_i(b_i^-) < \Psi_i(b_i)\). Therefore,
\[
\mathrm{OVS}_i(b_i^-) > \mathrm{OVS}_i(b_i).
\]
Since the optimistic virtual surplus for all other providers \(j \neq i\) remains unchanged, provider \(i\) is weakly more likely to be the argmax with the lower bid \(b_i^-\).
\end{proof}

The Claim 2 and Claim 3 prove that by submitting a lower bid \( b_i^- \), provider \( i \) ensures it will remain in the active set for at least as many rounds and stages as it would have with the higher bid \( b_i \).

During any exploitation step, the provider with the highest optimistic virtual surplus (OVS) is chosen. As \( \text{OVS}_i(b_i^-) > \text{OVS}_i(b_i) \) while \( \text{OVS}_j \) for \( j \neq i \) remains unchanged, provider \( i \) is strictly more competitive when it lowers its bid.

Since a lower bid guarantees that an provider remains eligible for selection for at least as long and makes it a more competitive choice in every selection round, the total number of allocations \( \mathcal{A}_i(b_i^-, T) \) must be greater than or equal to \( \mathcal{A}_i(b_i, T) \). This completes the proof.

\end{proof}

\subsection{Regret analysis of \oursupucb}
%%%%%%%%%%%%%%%%%%%%%%%%%%%%%%%%%%%%%%%%%%%%%%%%%%%%%%%%%%%%%%%%%%%%%%%%
In this subsection, we prove that \oursupucb\ incurs the regret $O(\sqrt{T})$. The proof is similar to SupLinUCB as our algorithm is derived from it. However, it has certain differences as we need to ensure monotonicity. Still, the regret dependency on $T$ remains the same, albeit, constant factors increase (for \oursupucb, dependency on number of providers is quadratic in $M$, whereas for SupLinUCB, it is $\sqrt{M})$.

The difference in the final regret proof is due to (i) rewards are also bid dependent, which is not the case with \cite{chu2011contextual}. (ii) We need monotonicity, and hence, our active sets retain the providers in the active set for more rounds than that in SupLinUCB.  (iii) We need to have an optimal reverse auction
(iv) To claim the regret guarantees, we need the following lemma.

\begin{lemma}
\label{lemma:regret}
With probability at least $1 - \kappa$, for any round $t \in [T]$ and any stage $s \in [S_{\max}]$, the following hold:

\begin{enumerate}
    \item For all $i \in [M]$,
    \[
    (\hat{v}_{\,i,t}^s + w_{i,t}^s) - \Psi_i(b_i) 
    \;\ge\;
    v_{i, t} - \Psi_i(b_i)
    \;\ge\;
    (\hat{v}_{\,i,t}^s - w_{i,t}^s) - \Psi_i(b_i).
    \]
    
    \item The optimal provider is never eliminated: $i_t^* \in \hat{A}_s$.
    
    \item For any $i \in \hat{A}_s$,
    \[
    \big(v_{i_t^*, t} - \Psi_{i_t^*}(b_{i_t^*})\big)
    - 
    \big(v_{i,t} - \Psi_i(b_i)\big)
    \le 2^{3-s}.
    \]
\end{enumerate}
\end{lemma}

We begin my stating some lemmas from other papers, that trivially follow for \oursupucb. 

%%%%%%%%%%%%%%%%%%%%%%%%%%%%%%

We begin with the following lemma from the seminal paper of UCB.
\begin{lemma}[Lemma 15, \cite{auer2003confidence}]
\label{lem:ucb_vit}
For any $t$ and any stage $s$,
\[
ucb_{i,t}^{s} - 2 \cdot w_{i, t}^{s} \le v_{i,t} \le ucb_{i,t}^{s} \quad \text{for any } i
\]
where, $ucb_{i,t}^{s} = \hat{v}_{i,t}^{s} + w_{i,t}^{s}$
\end{lemma}

The following lemmas from Chu et al.~\cite{chu2011contextual} for SupLinUCB are valid of \oursupucb.

\begin{lemma}[Lemma 2, \cite{chu2011contextual}]
For each $s \in [S]$ and $i \in \mathcal{N}$, suppose $\Lambda^s_{i,t+1} = \Lambda^s_{i,t} \cup \{t\}$. 
Then, eigenvalues of $A_{i,t}$ can be arranged so that $\lambda^j_{i,t} \le \lambda^j_{i,t+1}$, 
for all $j$ and
\[
s^2_{i,t} \le 10 \sum_{j=1}^{d} \frac{\lambda^j_{i,t+1} - \lambda^j_{i,t}}{\lambda^j_{i,t}}.
\]
\end{lemma}

\begin{lemma}[Lemma 3, \cite{chu2011contextual}]
Using notation in \textit{BaseLinUCB-S} and assuming $|\Lambda^s_{i,T+1}| \ge 2$, we have
\[
\sum_{t \in \Lambda^s_{i,T+1}} s_{i,t} \le 5 \sqrt{d |\Lambda^s_{i,T+1}| \ln |\Lambda^s_{i,T+1}|}.
\]
\end{lemma}

\begin{lemma}[Lemma 4, \cite{chu2011contextual}]
For each $s \in [S]$, each $t \in [T]$, and any fixed sequence of feature vectors $x_t$, 
with $t \in \Lambda^s_{I_t,t}$, the corresponding rewards $v_{I_t,t}$ are independent random variables such that 
$\mathbb{E}[v_{I_t,t}] = \theta_i^{\top} x_t$.
\end{lemma}

\begin{lemma}[Lemma 6, \cite{chu2011contextual}]
For all $s \in [S]$ and $i \in \mathcal{N}$,
\[
|\Lambda^{s}_{i,T+1}| \le 5 \cdot 2^{s}(1 + \alpha^2)\sqrt{d|\Lambda^{s}_{i,t+1}|}.
\]
\end{lemma}

\begin{proof}

\noindent\textbf{Part 1:}  

From~\ref{lem:ucb_vit}, we have
\[
\hat{v}_{i,t}^{s} + w_{i,t}^{s} - 2 \cdot w_{i, t}^{s} \le v_{i,t} \le \hat{v}_{i,t}^{s} + w_{i,t}^{s} \quad \text{for any } i
\]

Subtracting $\Psi_i(b_i)$, i.e.,
\[
\hat{v}_{i,t}^{s} + w_{i,t}^{s} - 2 \cdot w_{i, t}^{s} - \Psi_i(b_i)\le v_{i,t} - \Psi_i(b_i)\le \hat{v}_{i,t}^{s} + w_{i,t}^{s} - \Psi_i(b_i) \quad \text{for any } i
\]

\[
    (\hat{v}_{\,i,t}^s + w_{i,t}^s) - \Psi_i(b_i) 
    \;\ge\;
    v_{i, t} - \Psi_i(b_i)
    \;\ge\;
    (\hat{v}_{\,i,t}^s - w_{i,t}^s) - \Psi_i(b_i).
    \]

\medskip
\noindent\textbf{Part 2:}  
The lemma trivially holds for $s = 1$.  
For $s > 1$, we have $\hat{A}_s \subseteq \hat{A}_{s-1}$, and from the algorithm as it advances from stage $s-1$ to $s$, it is clear that
\[
    w_{i,t}^s \le 2^{-(s-1)} 
    \quad \text{and} \quad 
    w_{i_t^*(x_t),t}^s \le 2^{-(s-1)}.
\]

From part 1 of the the lemma we have
\[
    (\hat{v}_{\,i_t^*}^{s-1} + w_{i_t^*}^{s-1}) - \Psi_{i_t^*}(b_{i_t^*}) 
    \;\ge\;
    v_{i_t^*} - \Psi_{i_t^*}(b_{i_t^*}).
    \]
and for any $j \in \hat{A}_s$, using Part 1 of the lemma and the upper bound on $w_{i, t}^s$, we have
\[
    v_{i, t} - \Psi_i(b_i)
    \;\ge\;
    (\hat{v}_{\,i,t}^s + w_{i,t}^s) - \Psi_i(b_i) - 2 \cdot  w_{i,t}^s.
    \;\ge\;
    (\hat{v}_{\,i,t}^s + w_{i,t}^s) - \Psi_i(b_i) - 2 \cdot  2^{1-s}
    \]
From the definition,
\[
    v_{i_t^*(x_t),t}^s - \Psi_{i_t^*}(b_{i_t^*}) \ge v_{\,i,t}^s - \Psi_i(b_i).
\]
Using this and the above inequalities, we get: 
\[ (\hat{v}_{\,i_t^*}^s + w_{i_t^*}^s) - \Psi_{i_t^*}(b_{i_t^*}) \geq \max_{a \in \hat{A}_s} \big\{ (\hat{v}_{a,t}^s + w_{a,t}^s) - \Psi_a(b_a) \big\} - 2 \cdot 2^{1-s} \big\}\]
which is the Line[21] from the Algorithm~\ref{alg:suplinucb}.

Hence, the provider $i_t^*(x_t)$ will belong to $\hat{A}_s$ for all $s$ (i.e., it will never be eliminated for context $x_t$) due to the rule defined in Line[21], Algorithm~\ref{alg:suplinucb}.

\medskip
\noindent\textbf{Part 3:}  
From Line[21] of Algorithm~\ref{alg:suplinucb}, we have
\[ (\hat{v}_{i,t}^s + w_{i,t}^s) - \Psi_i(b_i) \geq (\hat{v}_{\,i_t^*}^s + w_{i_t^*}^s) - \Psi_{i_t^*}(b_{i_t^*}) - 2 \cdot 2^{1-s} \]

Using part 1 of the lemma, for provider $i$ we have: 
\[
(\hat{v}_{\,i,t}^s - w_{i,t}^s) - \Psi_i(b_i) \leq v_{i, t} - \Psi_i(b_i).
\]
Adding $2 \cdot w_{i,t}^s$ to each side
\[
(\hat{v}_{\,i,t}^s + w_{i,t}^s) - \Psi_i(b_i) \leq v_{i, t} - \Psi_i(b_i) + 2 \cdot w_{i,t}^s.
\]
Using part 1 of the lemma, for optimal provider $i_t^*$ we have: 
\[
    (\hat{v}_{\,i_t^*}^{s} + w_{i_t^*}^{s}) - \Psi_{i_t^*}(b_{i_t^*}) 
    \;\ge\;
    v_{i_t^*} - \Psi_{i_t^*}(b_{i_t^*}).
    \]

Substituting these values, 
\[ v_{i, t} - \Psi_i(b_i) + 2 \cdot w_{i,t}^s \geq v_{i_t^*} - \Psi_{i_t^*}(b_{i_t^*}) - 2 \cdot 2^{1-s} \]

Rearranging the terms, we get: 
\[ (v_{i_t^*} - \Psi_{i_t^*}(b_{i_t^*})) - (v_{i, t} - \Psi_i(b_i)) \leq 2 \cdot 2^{1-s} + 2 \cdot w_{i,t}^s\]

We can simplify the RHS as follows:
\begin{align*}
    2 \cdot 2^{1-s} + 2 \cdot w_{i,t}^s &= 2^{2-s} + 2\cdot 2^{1-s} \\
    &= 2 \cdot 2^{2-k} \\
    &= 2^{3-k}
\end{align*}

Hence, 
\[ (v_{i_t^*} - \Psi_{i_t^*}(b_{i_t^*})) - (v_{i, t} - \Psi_i(b_i)) \leq 2^{3-k} \]
\end{proof}

% MERGE THIS PROOF SKETCH 

% % \begin{proofsketch}
% The first part follows directly from Lemma~4 of~\citet{chu2011contextual} by subtracting the virtual cost term from both sides of the inequality.  
% For the second part, we use the confidence bound $w_{i,t}^s \le 2^{-(s-1)}$, which follows from the algorithm’s update rule. Combining this bound with Part~(1) implies that the optimal provider never satisfies the elimination criterion, and thus remains in the active set.  
% Finally, Part~(3) applies Part~(1) and the elimination condition to bound the instantaneous regret at round $t$.
% \end{proofsketch}

We restrict the learning during round-robin ordering only (Lines 10-13, Algorithm~\ref{alg:suplinucb}), to which we add an additional decision rule for provider selection (Line 27). 
Note that, this additional rule was not used in~\citet{chu2011contextual}. 
Hence, the main challenge is to bound the number of rounds of provider selection, which is done using this decision rule. (We refer to it as $\Lambda_{est}^s$ in our analysis.)
To this extent, we introduce some notations also mentioned in the algorithm: 

Let $\Lambda_0$ be the set of rounds in which the provider in pure exploitation phase was selected (Lines [14-17]). 
Let $\Lambda_{\mathrm{est}}^s$ be the set of rounds the provider was selected in exploitation in current phase (Lines [25-27]), and let $\Lambda_{T+1}^s = \bigcup_i \Lambda_{i,T+1}^s$.

The expression in Claim~\ref{claim} will be useful to bound the regret.
\begin{claim}\label{claim}
At each stage $s$, $|\Lambda^{s}_{\text{est}}| \le (n - 1) \cdot |\Lambda^{s}_{T+1}|$.
\end{claim}

\begin{proof}
Let's consider $M$ consecutive rounds for any stage $s$.
Assume selection of the provider is done as per Lines[25-27] for each of the $M$ rounds. 
Note that provider selection in this decision block is done if and only if there exists a provider $j$ such that $j \neq k$ (where $k$ is the designated provider for the round), and $w_{j,t}^s > 2^{-s}$. 
After $M$ consecutive rounds, each provider has received its designated round (as in Line[7]). 

Hence, if for some provider $k$, $w_{k,t}^s > 2^{-s}$, then this provider $k$ should be selected on its designated round. 

Assuming selection in Lines~[25--27] occurs for \(M\) consecutive rounds leads to a contradiction. Thus, at least one provider must be selected in its designated round within these \(M\) rounds. Hence, the provider is selected in Lines~[25--27] for at most \(M-1\) rounds. Until the condition in Line~[14] is satisfied, the claim holds at every stage \(s\).

\end{proof}

\begin{theorem}
\label{thm:regret}
\oursupucb\ has regret $\mathcal{O}\!\left(M^{2} \sqrt{dT \ln T}\right)$
with probability at least $1 - \kappa$, 
if it is run with 
$
\alpha = \sqrt{\tfrac{1}{2} \ln \tfrac{2TM}{\kappa}}.
$
\end{theorem}

\begin{proof}

As $2^{-S} \le 1 / \sqrt{T}$, we have
\[
    \{1, \ldots, T\} = \Lambda_0 \cup \bigcup_s \Lambda_{T+1}^s \cup \Lambda_{\mathrm{est}}^s.
\]

\medskip
\noindent
Now, the total regret can be written as
\[
\mathbb{R}_T(\mathcal{ALG}) = \sum_{t=1}^{T} \left[ \left(\theta_{i_t^*}^\top x_t - \psi \left(b_{i_t^*}\right) \right) - \left(\theta_{I_t}^\top x_t - \psi \left(b_{I_t}\right) \right) \right],
\]

Rearranging the summation by partitioning into $\Lambda_0$, $\Lambda_{T+1}^s$, and $\Lambda_{\mathrm{est}}^s$, we get:
\begin{align*}
\mathbb{R}_T (\mathcal{ALG})
&= \sum_{t \in \Lambda_0} 
    \left[ \left(\theta_{i_t^*}^\top x_t - \psi \left(b_{i_t^*}\right) \right) - \left(\theta_{I_t}^\top x_t - \psi \left(b_{I_t}\right) \right) \right] \\
&\quad + 
    \sum_{s=1}^S 
    \left[
        \sum_{t \in \Lambda_{T+1}^s}
            \left[ \left(\theta_{i_t^*}^\top x_t - \psi \left(b_{i_t^*}\right) \right) - \left(\theta_{I_t}^\top x_t - \psi \left(b_{I_t}\right) \right) \right]
    \right] \\
&\quad +
    \sum_{t \in \Lambda_{\mathrm{est}}^s}
        \left[ \left(\theta_{i_t^*}^\top x_t - \psi \left(b_{i_t^*}\right) \right) - \left(\theta_{I_t}^\top x_t - \psi \left(b_{I_t}\right) \right) \right].
\end{align*}

Using the Claim and Lemma results, we get
\begin{align*}
\mathbb{R}_T 
&\le \frac{2}{\sqrt{T}} |\Lambda_0|
    + \sum_{s=1}^S M \sum_{t \in \Lambda_{T+1}^s}
        \left[ \left(\theta_{i_t^*}^\top x_t - \psi \left(b_{i_t^*}\right) \right) - \left(\theta_{I_t}^\top x_t - \psi \left(b_{I_t}\right) \right) \right] \\
&\le \frac{2}{\sqrt{T}} |\Lambda_0| 
    + M \sum_{s=1}^S 
        \sum_{i=1}^{M} 
            8 \cdot 2^{-s} \cdot |\Lambda_{i,T+1}^s| \\
&\le \frac{2}{\sqrt{T}} |\Lambda_0|
    + M \sum_{i=1}^{M} 
        \sum_{s=1}^S 
            40 \cdot (1 + \ln(2TM/\kappa)) 
            \cdot \sqrt{d \, |\Lambda_{i,T+1}^s|} \\
&\le \frac{2}{\sqrt{T}} |\Lambda_0|
    + M^2 \cdot 40 \cdot (1 + \ln(2TM/\kappa)) \cdot \sqrt{STd}.
\end{align*}

Hence,
\[
\mathbb{R}_T \le 
2 \sqrt{T} + 40 M^2 (1 + \ln(2TM/\kappa)) \cdot \sqrt{STd}.
\]
\end{proof}

% MERGE THE POOF AND PROOF SKETCH 

% \begin{proofsketch}
% The proof follows the structure of Theorem~6 in~\citet{auer2003confidence}, with additional terms introduced to account for differences. Primarily, as our algorithm introduces two new components: round-robin exploration and within-stage exploitation—to ensure monotone MAB allocation, {we need to show that with high probablity, active sets $\hat{A}_s$s cannot eliminate an optimal provider (proved in Lemma~\ref{lemma:regret}).  We  need to bound that even if $\hat{A}_s$s contain more agents than that in SupLinUCB at any round, (a) the number of times bad sub-optimal providers are used is bounded by constant with high probability. (b) the sub-optimal arms in active sets that may be pulled more frequently, but the regret incurred due to such pulls is bounded.}
% Claim~\ref{claim} provides an upper bound on how many times a bad sub-optimal provider is used during the within-stage exploitation phase.  Lemma~\ref{lemma:regret} provides reward guarantees of a sub-optimal arm in the active set. 
% Finally, the cumulative regret is obtained by summing the instantaneous regret across all selection rounds, and applying {Lemma~\ref{lemma:regret}} together with Claim~\ref{claim}  completes the bound. 
% \end{proofsketch}

From Theorems~\ref{thm:opt_auc},~\ref{thm:truthful_epic_epir},~\ref{thm:mono}, and~\ref{thm:regret}, the reward structure used for arms in \oursupucb\ and the definition of \ourmabalgo, we conclude,
\begin{corollary}
    \ourmabalgo, is EPIC, EPIR, optimal reverse auction that learns $\theta_i$s with regret guarantee of $O(\sqrt{T})$.
\end{corollary}

%%%%%%%%%%%%%%%%%%%%%%%%%%%%%%%%%%%%%%%%%%%%%%%%%%%%%%%%%%%%%%%%%%%%%%%%
\section{Experimental Analysis}
\label{sec:exp_ana}

\paragraph{Simulation environment}
Our experimental study instantiates \ourmabalgo. For each query $q_t, t \in [T]$, we sample a $d=5$ dimensional Gaussian context $x_t \sim \mathcal{N}(\mathbf{0}, \Sigma)$, where $\Sigma$ blends a diagonal scale of $0.2$ with off-diagonal correlations of $0.05$. Four strategic providers submit bids drawn from provider-specific log-normal virtual-value models: for provider $i$, the latent parameters $(\mu_i, \sigma_i)$ depend affinely on $x_t$ through a learned context matrix and a scale factor of $0.15$. 

The Gaussian reward model induces light-tailed, mean-zero perturbations around the linear signal 
$\langle \theta_i, x_t \rangle$, keeping deviations symmetric and controlled at the 
$\sigma_i$ scale. 
To stress our mechanism with asymmetric noise, we also study an exponential sampling scheme that first maps the linear score through 
\[
\lambda_{i,t} = \operatorname{softplus}(\langle \theta_i, x_t \rangle)
\]
to ensure positivity, and then draws 
\[
R_{i,t} \sim \operatorname{Exp}(\lambda_{i,t}).
\]
Here, the reward’s mean is $\lambda_{i,t}^{-1}$ and its variance $\lambda_{i,t}^{-2}$, 
yielding strictly non-negative observations punctuated by rare, high-magnitude outcomes. 
Contrasting regret and revenue under these two noise regimes, therefore, probes how 
\textsc{SupLinUCB-S-OPT} copes with centered, light-tailed fluctuations versus skewed, 
heavy-tailed shocks that better emulate sporadic windfall bids. 

We fix the exploration parameter of \textsc{REV-SupLinUCB-S-OPT} at $\alpha = 0.75$ and regularize the Gram matrices with $\lambda = 1.0$. Unless noted otherwise, we operate at horizon $T = 100K$, and also verify scalability up to $T = 10^6$ courtesy of incremental Sherman--Morrison updates that sustain per-stage Gram matrices in $\mathcal{O}(d^2)$ time per round.

\paragraph{Protocol and metrics}
 To obtain statistically stable estimates, we repeat each configuration over $(N = 40)$ independent random seeds. For run $n$, we generate fresh pseudo-random draws for contexts and bids while keeping the structural parameters fixed. We summarize performance by plotting the averaged cumulative realized regret alongside a scaled $(\sqrt{T})$ reference curve, which highlights the sublinear growth rate and the steady decline of per-round regret. 
 We also measure the realized regret in each round, illustrating its convergence toward zero and the reduced dispersion once forced exploration subsides. 
 Finally, we chart the cumulative actual revenue next to the clairvoyant benchmark, showing how the learned policy progressively narrows the optimality gap while sustaining upside in later rounds.

\paragraph{Implementation details}
The simulator is implemented in Python 3.13 with \texttt{numpy}, \texttt{pandas}, and \texttt{matplotlib}. All experiments execute on a single workstation (Windows 11, 32\,GB RAM); a 40-run sweep at $T = 10K$ completes in roughly 45s thanks to the cached Gram inverses. 
\begin{figure*}[htbp]
  \centering
  \begin{minipage}[t]{0.32\textwidth}
    \centering
    \includegraphics[width=\linewidth]{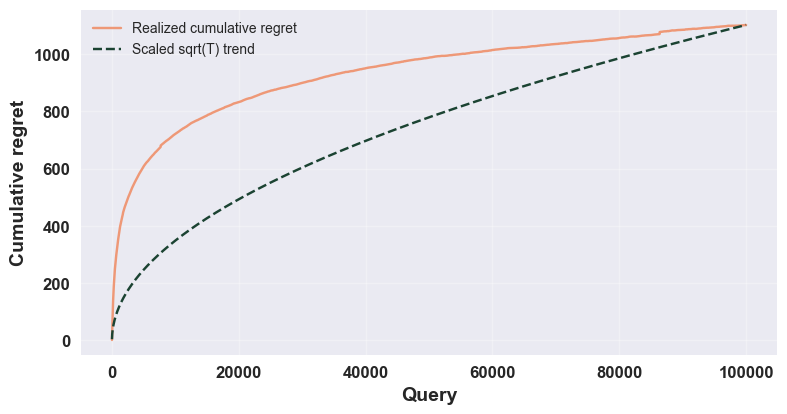}
    \caption*{(a) Averaged cumulative realized regret with scaled $\sqrt{T}$ reference.}
  \end{minipage}\hfill
  \begin{minipage}[t]{0.32\textwidth}
    \centering
    \includegraphics[width=\linewidth]{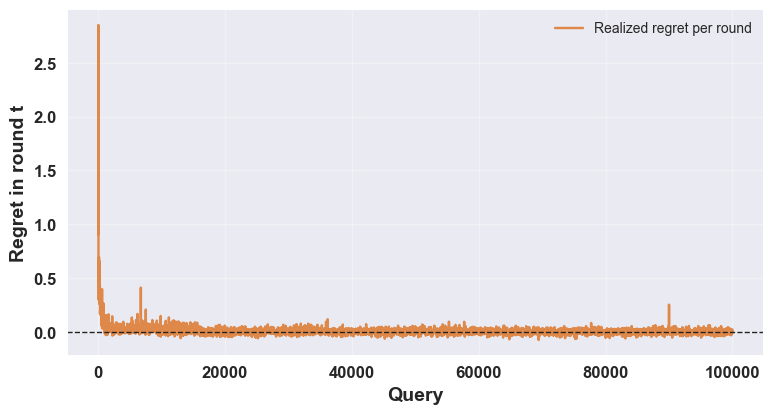}
    \caption*{(b) Realized regret per round showing convergence toward zero.}
  \end{minipage}\hfill
  \begin{minipage}[t]{0.32\textwidth}
    \centering
    \includegraphics[width=\linewidth]{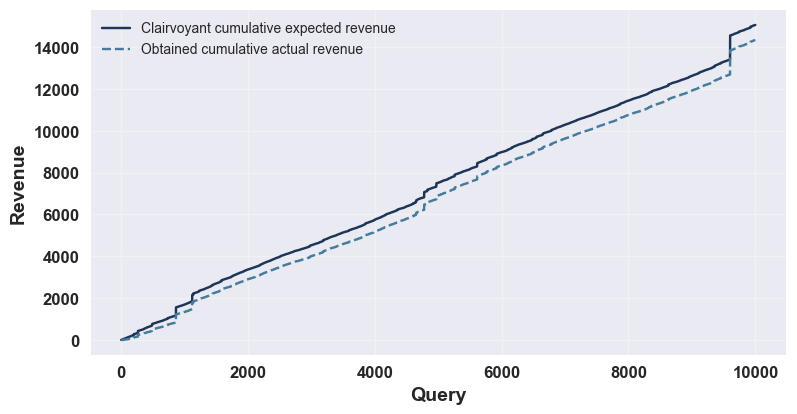}
    \caption*{(c) Cumulative actual revenue vs.\ clairvoyant benchmark.}
  \end{minipage}
  \caption{Regret and revenue curve for \ourmabalgo\ averaged across $N=40$ random seeds.}
  \label{fig:rev_diagnostics}
\end{figure*}

\begin{figure*}[htbp]
  \centering
  \begin{minipage}[t]{0.32\textwidth}
    \centering
    \includegraphics[width=\linewidth]{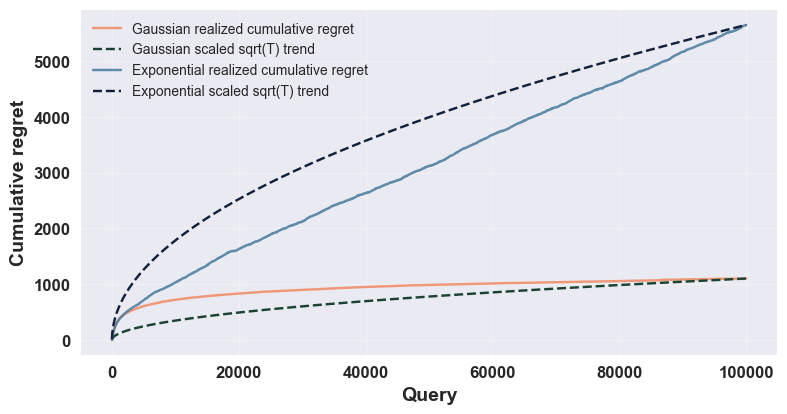}
    \caption*{(a) Averaged cumulative realized regret with scaled $\sqrt{T}$ reference.}
  \end{minipage}\hfill
  \begin{minipage}[t]{0.32\textwidth}
    \centering
    \includegraphics[width=\linewidth]{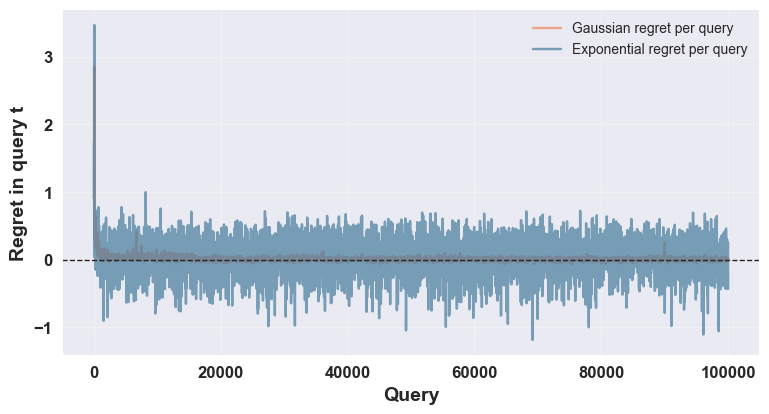}
    \caption*{(b) Realized regret per round showing convergence toward zero.}
  \end{minipage}\hfill
  \begin{minipage}[t]{0.32\textwidth}
    \centering
    \includegraphics[width=\linewidth]{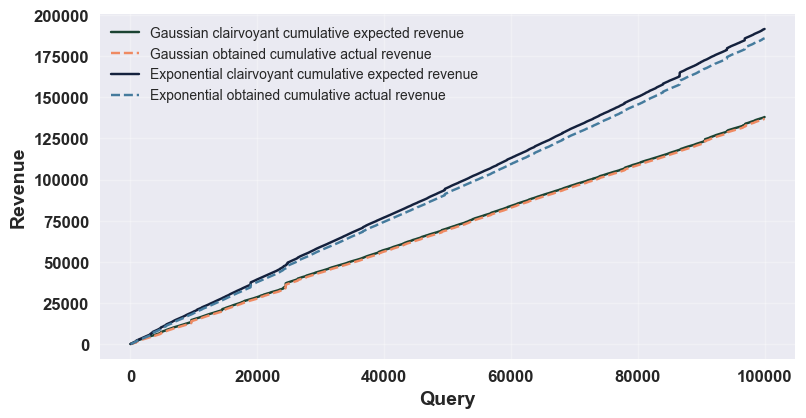}
    \caption*{(c) Cumulative actual revenue vs.\ clairvoyant benchmark.}
  \end{minipage}
  \caption{Regret and revenue curve for \ourmabalgo\ averaged across $N=40$ random seeds.}
  \label{fig:rev_diagnostics_mix}
\end{figure*}

\paragraph{Results}
 We evaluate \ourmabalgo\ over $N = 40$ independent random seeds and summarize the averaged trajectories in Figure~\ref{fig:rev_diagnostics}. In the figure, Panel (a) plots the cumulative realized regret together with a scaled $\sqrt{t}$ guide curve, illustrating the expected sublinear growth and confirming that the algorithm steadily reduces the gap relative to the clairvoyant benchmark. Panel (b) reports the per-round realized regret, which contracts toward zero once the forced-exploration stages conclude, indicating rapid convergence to high-quality selections. Panel (c) compares the cumulative actual revenue with the clairvoyant revenue; the two curves tighten over time, showing that the policy captures most of the attainable surplus while retaining upside during the later exploitation rounds. 

 Figure~\ref{fig:rev_diagnostics_mix} plots the results for Exponential distribution, and we identical trends vis-a-vis the Gaussian distribution.

%%%%%%%%%%%%%%%%%%%%%%%%%%%%%%%%%%%%%%%%%%%%%%%%%%%%%%%%%%%%%%%%%%%%%%%%%%%%%%%%%%%%%%%%%%%%%%%%%%%%%%%%%%%%%%%%%%%%%%%%%%%%%%%%%%%%%%%%%%%%%%%%%%%%%%%%%%%%%%%%%%%%%%%%%%%%%%%%%%%%%%%%%%%%%%%%%%%%%%%%%%%%%%%%%%%%%%%%%%%%%%%%%%%%%%%%%%%%%%%%%%%%%%%%%%%%%%%%%%%%%%%%%%%%%%%%%%%%%%%%%%%%%%%%

%%%%%%%%%%%%%%%%%%%%%%%%%%%%%%%%%%%%%%%%%%%%%%%%%%%%%%%%%%%%%%%%%%%%%%%%
%%%%%%%%%%%%%%%%%%%%%%%%%%%%%%%%%%%%%%%%%%%%%%%%%%%%%%%%%%%%%%%%%%%%%%%%
%%%%%%%%%%%%%%%%%%%%%%%%%%%%%%%%%%%%%%%%%%%%%%%%%%%%%%%%%%%%%%%%%%%%%%%%%%%%%%%%%%%%%%%%%%%%%%%%%%%%%%%%%%%%%%%%%%%%%%%%%%%%%%%%%%%%%%%%%%%%%%%%

%%%%%%%%%%%%%%%%%%%%%%%%%%%%%%%%%%%%%%%%%%%%%%%%%%%%%%%%%%%%%%%%%%%%%%%%

\section{Conclusion}
\label{sec:conclusion}
We introduced the first truthful optimal reverse  contextual MAB mechanism for adaptive LLM model selection, integrating provider-side cost elicitation with user-side learning. Our randomized resampling  procedure, \ourrealgo, extends prior forward-auction frameworks to reverse, cost-minimizing settings while preserving truthfulness and sublinear regret. The resulting \ourmabalgo\ mechanism aligns incentives between users and LLM providers and achieves $O(\sqrt{T})$ regret. More broadly, our work establishes a foundation for mechanism design in multi-model AI ecosystems, where learning and economic incentives must co-evolve.

\paragraph{Future Work}
Future extensions could incorporate minimum thresholds on the costs reported by LLM providers. Another promising direction is to handle uncertainty in query token requirements and user-side budget limits. Incorporating such practical and external constraints would enhance the applicability of our reverse optimal contextual MAB framework to real LLM marketplaces.

%%%%%%%%%%%%%%%%%%%%%%%%%%%%%%%%%%%%%%%%%%%%%%%%%%%%%%%%%%%%%%%%%%%%%%%%

\bibliographystyle{unsrt}

\end{document}